\documentclass[aps,prx,twocolumn,superscriptaddress,nofootinbib,floatfix]{revtex4-1}

\usepackage{xcolor}
\definecolor{darkblue}{RGB}{40, 85, 120} 
\definecolor{quantumviolet}{HTML}{53257F} 
\definecolor{quantumgray}{HTML}{555555} 
\definecolor{quantumgreen}{HTML}{007474} 
\definecolor{quantumblue}{HTML}{3A5FCD} 
\definecolor{quantumpurple}{HTML}{8A2BE2} 
\definecolor{warmorange}{HTML}{FF8C42} 
\definecolor{violetpurple}{HTML}{8A2BE2} 
\definecolor{teal}{HTML}{008080} 
\definecolor{softbeige}{HTML}{ECE5D7} 
\definecolor{charcoalgray}{HTML}{2E2E2E} 
\definecolor{coralred}{HTML}{FF6F61} 
\definecolor{brightyellow}{HTML}{FFD700} 
\definecolor{cyanblue}{HTML}{1C77C3} 
\definecolor{deepbluegray}{HTML}{2C3E99} 
\definecolor{softblue}{HTML}{AFCBFF} 

\usepackage{amsmath}
\usepackage{amssymb}
\usepackage{amsthm}
\usepackage{bm}
\usepackage{braket}
\usepackage{xcolor}
\usepackage{graphicx}
\usepackage[caption=false]{subfig}
\usepackage[colorlinks, linktocpage, hypertexnames = false]{hyperref}

\usepackage{dsfont}

\hypersetup{linkcolor = {gray!80!black},
citecolor = {blue!90!black}, 
urlcolor = {gray!60!black}}

\usepackage{tikz}
\usetikzlibrary{positioning,intersections}
\usetikzlibrary{calc}
\usetikzlibrary{shapes.geometric}
\usepackage{braids}
\usepackage{tikz-cd}
\usetikzlibrary{shapes.geometric,shapes.misc}
\usetikzlibrary{arrows,matrix,calc,scopes,decorations.markings,snakes}
\usetikzlibrary{arrows.meta}
\usetikzlibrary{intersections, patterns,fit} 
\usetikzlibrary{decorations.pathreplacing,angles,quotes}

\usepackage{qcircuit}

\theoremstyle{theorem}
\newtheorem{theorem}{Theorem}
\newtheorem{lemma}{Lemma}
\newtheorem{definition}{Definition}
\newtheorem{corollary}{Corollary}

\theoremstyle{remark}

\newcommand{\Herm}{\mathsf{Herm}}

\newcommand\Tr{\operatorname{Tr}}

\usepackage{euscript}

\usepackage{microtype}  

\usepackage[bbgreekl]{mathbbol} 
\DeclareMathAlphabet{\mathcal}{OMS}{cmsy}{m}{n}

\usepackage[noframe]{showframe}
\usepackage{framed}
\usepackage{lipsum}
 {\endMakeFramed}
\definecolor{shadecolor}{gray}{0.9} 

\definecolor{revisionblue}{RGB}{0,65,160}
\newenvironment{revision}{\begingroup\ifdefined\showrevision\color{revisionblue}\fi}{\endgroup}

\makeatletter
\newsavebox{\@brx}
\newcommand{\llangle}[1][]{\savebox{\@brx}{\(\m@th{#1\langle}\)}%
  \mathopen{\copy\@brx\kern-0.5\wd\@brx\usebox{\@brx}}}
\newcommand{\rrangle}[1][]{\savebox{\@brx}{\(\m@th{#1\rangle}\)}%
  \mathclose{\copy\@brx\kern-0.5\wd\@brx\usebox{\@brx}}}
\makeatother

\usepackage{multirow}

\usepackage{orcidlink} 
\usepackage{adjustbox}

\usepackage{graphicx}
\usepackage{enumerate}
\usepackage{epstopdf}
\usepackage{amsmath,amsthm,amssymb,epsfig,amsfonts,mathrsfs}

\hypersetup{colorlinks=true}

\def\Tr{\mathop{\mathrm{Tr}}\nolimits}
\newcommand{\QR}{\overrightarrow{Q}_{\rm KD}}

\begin{document}

\flushbottom

\title{Temporal State Tomography via Quantum Snapshotting the Temporal Quasiprobabilities}

\author{Zhian Jia\orcidlink{0000-0001-8588-173X}}
\email{giannjia@foxmail.com}

\affiliation{Institute of Quantum Physics,  School of Physics, Central South University,
Changsha 410012, China}

\date{\today}

\begin{abstract}
Quantum tomography is a cornerstone of quantum information science, enabling the reconstruction of states and channels from experimental data. Here we introduce a new paradigm, temporal state tomography (TST), for reconstructing quantum processes across multiple times. Our approach is based on temporal quasiprobability distributions (TQDs), which, in the informationally complete setting, provide a complete description of multi-time quantum processes and uniquely determine temporal states. We formulate TST as a unified framework for reconstructing both density operators and quantum channels within a single scheme. We show that any TQD can be obtained via classical post-processing of measurement outcomes generated by a fixed set of quantum instruments, thereby establishing a direct operational route to accessing TQDs experimentally. For informationally complete TQDs, the associated temporal state can be reconstructed via a temporal Bloch-type representation. Leveraging this correspondence, we derive the sample complexity of TST, thereby quantifying its statistical efficiency.
\end{abstract}

\maketitle


\emph{Introduction.} ---
A fundamental problem in quantum information science is the efficient extraction of information from unknown quantum objects—such as quantum states, quantum channels, quantum multi-time processes, and higher-order quantum operations—via a procedure known as quantum tomography, see, e.g.,\cite{Anshu2024,Mohseni2008processtom,Li2020hamiltoniantom,White2022processtom,Antesberger2024highertomography,Li2025combtomography}.
In standard formalism of quantum theory, the quantum state and quantum channel are characterized via density operator and completely positive trace-preserving (CPTP) maps respectively, which makes their tomography procedure works differently.

Temporal states provide a unified formalism encompassing quantum state tomography, quantum process tomography, and their extension to multi-time quantum processes. Within this framework, a tensor-product structure is introduced across distinct temporal instances, placing spatial and temporal degrees of freedom on an equal footing, in analogy with relativity.
The concept of temporal states dates back to 1984, with the consistent histories formulation of quantum theory~\cite{griffiths1984consistent} serving as a prototypical example. The central idea is to endow time with a tensor-product structure, $\mathcal{H}_{t_n}\otimes\cdots \otimes \mathcal{H}_{t_0}$. More generally, one may consider a spatiotemporal Hilbert space $\bigotimes_{(x,t)}\mathcal{H}_{(x,t)}$, within which spatial and temporal degrees of freedom are treated on an equal footing.

In recent years, a variety of temporal-state formalisms have been developed, see, e.g.,~\cite{fitzsimons2015quantum,fullwood2022quantum,Parzygnat2023pdo,Song2024causal,Lie2025stateovertime,gutoski2007toward,Chiribella2009comb,Pollock2018processtensor,oreshkov2012quantum,Aharonov2009multi,cotler2018superdensity,jia2024spatiotemporal}. The relationships and distinctions among these approaches have been analyzed~\cite{liu2023unification,Parzygnat2023pdo,Jia2025TemporalKirkwoodDirac,jia2026unifiedspatiotemporalquantumstates}.
In this work, by a temporal state we mean an object that encodes both the states at different time instances and the dynamical evolution between them. We further require temporal states to satisfy a quantum Kolmogorov consistency condition: for any temporal subset $\{ t_{i_1}, \dots, t_{i_k}\} \subseteq \{t_0, t_1, \dots, t_n\}$, the corresponding temporal state $\Upsilon_{t_{i_k}\cdots t_{i_1}}$ is obtained as a reduced state of $\Upsilon_{t_n \cdots t_0}$~\cite{Jia2025TemporalKirkwoodDirac}.

Consider a Markov multi-time quantum process $\mathfrak{P} = \bigl(\rho_{t_0}, \mathcal{E}_{t_1 \leftarrow t_0}, \dots, \mathcal{E}_{t_n \leftarrow t_{n-1}}\bigr)$, where $\rho_{t_0}$ is the initial state and each $\mathcal{E}_{t_j \leftarrow t_{j-1}}$ is a completely positive trace-preserving (CPTP) map describing the evolution from $t_{j-1}$ to $t_j$. Temporal states (in particular those arising from temporal Bloch tomography) can often be written as~\cite{fullwood2022quantum,Parzygnat2023pdo,Lie2025stateovertime,lie2025probingquantumstatesspacetime,fullwood2023quantum,Liu2025PDO,Jia2025TemporalKirkwoodDirac,jia2026unifiedspatiotemporalquantumstates}
\begin{equation}\label{eq:TemporalStateExp}
\Upsilon_{t_n \cdots t_0} = \mathcal{E}_{t_n \leftarrow t_{n-1}} \star_{\mathrm{TS}} \bigl(\cdots \star_{\mathrm{TS}} (\mathcal{E}_{t_1 \leftarrow t_0} \star_{\mathrm{TS}} \rho_{t_0})\bigr),
\end{equation}
where $\star_{\mathrm{TS}}$ denotes the temporal link product, whose precise form depends on the chosen formalism.

Quantum measurements are inherently probabilistic and, in general, provide only partial information about the underlying temporal quantum state. Consequently, obtaining a complete classical description of an unknown temporal state requires combining data from multiple suitably chosen measurements, we call this task the \emph{temporal state tomography} (TST).

\begin{figure}
    \centering
    \includegraphics[width=0.9\linewidth]{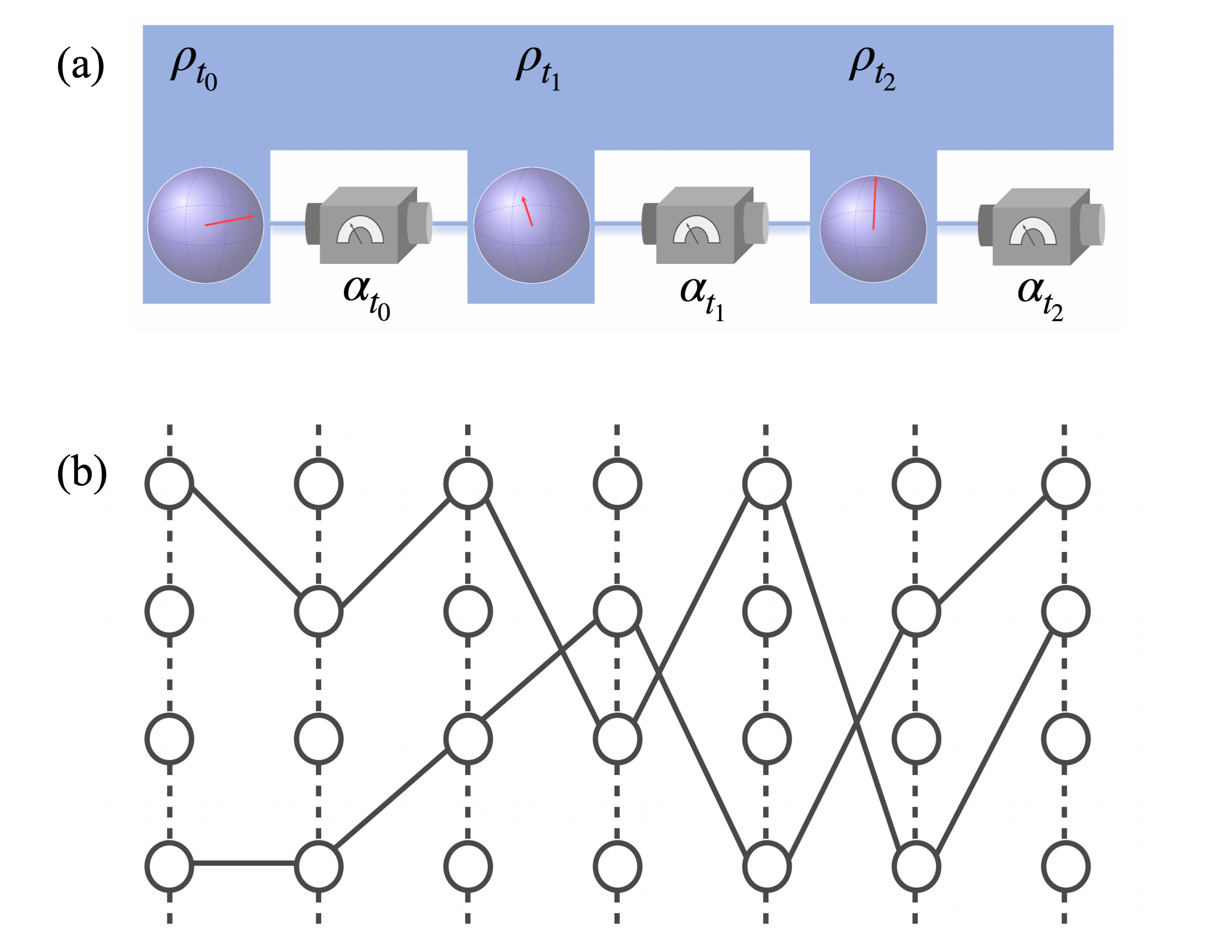}
\caption{(a) Temporal quantum states generalize the multipartite density operator formalism to the time domain, enabling a unified description of quantum systems with both timelike and spacelike correlations.(b) By selecting an appropriate phase space at each time step, one can construct a temporal phase space; each quantum trajectory within this temporal phase space is associated with a corresponding temporal quasiprobability distribution that encodes temporal quantum correlations.}
    \label{fig:TemporalState}
\end{figure}

\begin{definition}[Temporal State Tomography (TST)]\label{def:TST}
Let $\mathsf{TS}(\mathcal{H}_{t_n}\otimes\cdots\otimes \mathcal{H}_{t_0})$ denote the set of all temporal states for an $(n+1)$-step quantum temporal process.  
Consider $\delta,\varepsilon \in (0,1)$ and $N \in \mathbb{N}$, and let $\Upsilon \in \mathbf{TS}(\mathcal{H}_{t_n}\otimes\cdots\otimes \mathcal{H}_{t_0})$ be an unknown temporal state. Given $N$ copies of $\Upsilon$, the goal is to construct a classical description of a temporal state $\Upsilon'$ such that
\begin{equation}
    \mathbb{P}\big[d(\Upsilon,\Upsilon') \leq \varepsilon \big] \geq 1 - \delta.
\end{equation}
In other words, with probability at least $1 - \delta$, the distance between $\Upsilon$ and $\Upsilon'$ does not exceed $\varepsilon$. Here, $\varepsilon$ is referred to as the \emph{error} with respect to the distance function $d$, and $\delta$ is the \emph{failure probability}.
For the single-step case, this definition reduces to the standard quantum state tomography.
\end{definition}

However, unlike quantum state and process tomography, where positive semidefiniteness significantly simplifies reconstruction, temporal states are generally not positive semidefinite, and their state space is correspondingly more intricate. To address this, we employ the temporal quasiprobability distribution (TQD), recently introduced in~\cite{Jia2025TemporalKirkwoodDirac}. Some of its nonclassical properties are further discussed in~\cite{ding2026quantifying}; see also Refs.~\cite{fullwood2025spatiotemporalbornrule,Lie2025stateovertime,lie2025probingquantumstatesspacetime} for discussions of special cases.
Quasiprobabilistic formulation of quantum mechanics, originating from Wigner's seminal work~\cite{Wigner1932on}, is a powerful framework for capturing spatial nonclassicality; the TQD serves as its temporal counterpart. Building on a temporal generalization of the Kirkwood-Dirac quasiprobability distribution~\cite{Kirkwood1933quantum,Dirac1945on,margenau1961correlation,ArvidssonShukur2024KDreview}, it has been shown that the temporal state can be fully reconstructed from the TQD~\cite{Jia2025TemporalKirkwoodDirac,fullwood2025spatiotemporalbornrule,Lie2025stateovertime,lie2025probingquantumstatesspacetime}.

Analogous to the spatial case, the TQD is defined over a temporal phase space (see Fig.~\ref{fig:TemporalState}), assigning a quasiprobability to each quantum trajectory in this space. It captures temporal quantum correlations and dynamical information of the quantum multi-time process. For a fixed time instance $t_k$, the corresponding marginal TQD yields a quasiprobabilistic representation of $\rho_{t_k}$. In this sense, conventional spatial quasiprobability distributions arise naturally as marginals of the TQD. Moreover, the TQD satisfies a Kolmogorov consistency condition~\cite{Jia2025TemporalKirkwoodDirac}.
Since the TQD is generally complex-valued or can take negative values, one must develop physical schemes to access it experimentally. A standard approach is based on interferometric measurement schemes~\cite{Jia2025TemporalKirkwoodDirac,Lostaglio2023kirkwooddirac,Gherardini2024KD}; however, these are not well suited for TST. 
In this work, inspired by Ref.~\cite{Wang2024KD}, we introduce a quantum snapshotting scheme that transforms the TQD into a physically realizable sequential measurement implemented via a quantum instrument. The resulting temporal probability distribution is then classically postprocessed to reconstruct the desired TQD, which is subsequently used to perform TST task.

\vspace{1em}
\emph{Quantum multi-time processes and temporal quasiprobability distributions.} ---
A quantum multi-time process is specified by
\(
\mathfrak{P} = \bigl(\rho_{t_0}, \mathcal{E}_{t_1 \leftarrow t_0}, \cdots, \mathcal{E}_{t_n \leftarrow t_{n-1}}\bigr),
\)
where $\rho_{t_0}$ is the initial state and $\{\mathcal{E}_{t_k \leftarrow t_{k-1}}\}_{k=1}^n$ are time-ordered CPTP maps describing the evolution between successive times $t_{k-1}$ and $t_k$.
In this work, we primarily focus on the Markovian setting. For non-Markovian dynamics (including an explicit environment and the Stinespring dilation of the evolution map), the TQD formalism remains applicable; see Refs.~\cite{Jia2025TemporalKirkwoodDirac,jia2026unifiedspatiotemporalquantumstates,Kelvin2025KD}.

For a quantum multi-time process, we define the TQD as
\begin{equation}
Q(\beta_n, \dots, \beta_0) 
=
\operatorname{Tr} \Big[ \mathcal{P}^{t_n}_{\beta_n} \circ \mathcal{E}_{t_n \leftarrow t_{n-1}} \circ \cdots \circ \mathcal{P}^{t_1}_{\beta_1} \circ \mathcal{E}_{t_1 \leftarrow t_0} \circ \mathcal{P}^{t_0}_{\beta_0} (\rho_{t_0}) \Big],
\label{eq:temporal_quasiprobability}
\end{equation}
where $\beta_k$ denotes a phase-space point at time $t_k$, and $\mathcal{P}^{t_k}_{\beta_k}$ is the corresponding phase-space operation, which is generally not completely positive. Typical choices~\cite{Jia2025TemporalKirkwoodDirac,Wang2024KD} include the right projection $\mathcal{P}^{t_k}_{b_k}(\bullet) = \bullet\, \Pi_{b_k}$, the left projection $\mathcal{P}^{t_k}_{a_k}(\bullet) = \Pi_{a_k}\, \bullet$, and the doubled projection $\mathcal{P}^{t_k}_{a_k,b_k}(\bullet) = \Pi_{a_k}\, \bullet\, \Pi_{b_k}$, where $\{\ket{a_k}\}$ and $\{\ket{b_k}\}$ are orthonormal bases of the Hilbert space. To render the TQD informationally complete, these projectors can be generalized to an informationally complete POVM (IC-POVM) $\{K_{\beta}\}$. We refer to the resulting constructions as the left, right, and doubled IC-TQD, denoted by $\overleftarrow{Q}$, $\overrightarrow{Q}$, and $\overleftrightarrow{Q}$, respectively.
These TQDs are generally complex-valued and are therefore referred to as temporal Kirkwood–Dirac distributions \cite{Kirkwood1933quantum,Dirac1945on,Jia2025TemporalKirkwoodDirac,ding2026quantifying,ding2026quantifying,jia2026unifiedspatiotemporalquantumstates}. Their real part defines a quasiprobability distribution, known as the temporal Margenau–Hill  distribution \cite{margenau1961correlation,Jia2025TemporalKirkwoodDirac,fullwood2025spatiotemporalbornrule,Lie2025stateovertime,lie2025probingquantumstatesspacetime}.

These IC-TQDs can be regarded as quasiprobabilistic temporal states, equivalent to the temporal state in operator form~\cite{Jia2025TemporalKirkwoodDirac,jia2026unifiedspatiotemporalquantumstates}. To see this, consider the right Kirkwood–Dirac TQD:
\begin{equation}
\begin{aligned}
       & \overrightarrow{Q}_{\rm KD}(\beta_n, \ldots, \beta_0) \\
        = &\ \operatorname{Tr}\Big[ 
    \mathcal{E}_{t_n \leftarrow t_{n-1}} \Big( 
    \cdots 
    \mathcal{E}_{t_2 \leftarrow t_1} \big(
    \mathcal{E}_{t_1 \leftarrow t_0} ( 
    \rho_{t_0} K_{\beta_0} ) 
    K_{\beta_1} \big)  K_{\beta_2}  \cdots 
    \Big) 
    K_{\beta_n}
    \Big],
\end{aligned}
\label{eq:right_KD_TQD}
\end{equation}
where the Pauli measurements at the $k$-th step can be expanded as linear combinations of $K_{\beta_k}$. From this expansion, one obtains the temporal Pauli correlators $\overrightarrow{T}^{\mu_n,\ldots,\mu_0}=\langle \{\sigma_{\mu_n}\,\cdots,\sigma_{\mu_0}\}\rangle$ for all $\mu_n,\cdots,\mu_0=0,\cdots,3$, which determine the temporal state via the temporal Bloch formula:
\begin{equation}
    \overrightarrow{\Upsilon} = \frac{1}{2^{n+1}} \sum \overrightarrow{T}^{\mu_n \cdots \mu_0} \sigma_{\mu_n}\otimes \cdots \otimes \sigma_{\mu_0}.
\label{eq:temporal_state}
\end{equation}
For general qudit, we need to use the generalized Gell-Mann basis, and the overall normalization factor becomes $1/d^{n+1}$ \cite{Jia2025TemporalKirkwoodDirac,jia2026unifiedspatiotemporalquantumstates}.
Conversely, the TQD can be obtained from the temporal state operator via a generalized Born rule,
\begin{equation}
    \overrightarrow{Q}_{\rm KD}(\beta_n, \ldots, \beta_0)
    =
    \operatorname{Tr}\!\left( \overset{\rightarrow}{\Upsilon}\, K_{\vec{\beta}} \right),
\label{eq:TQD_from_temporal_state}
\end{equation}
where we abbreviate $K_{\beta_n}\otimes \cdots \otimes K_{\beta_0}$ as $K_{\vec{\beta}}$ (for doubled case, the corresponding temporal state becomes a doubled density operator; the temporal Born rule is more subtle \cite{jia2024spatiotemporal} yet still retains an inner-product structure.). Introducing the dual frame $G_{\vec{\beta}} := G_{\beta_n}\otimes \cdots \otimes G_{\beta_0}$, it is easy to prove that temporal state is given by
\begin{equation}\label{eq:TSdual}
    \overrightarrow{\Upsilon} = \sum_{\vec{\beta}} G_{\vec{\beta}}\, \overrightarrow{Q}_{\rm KD}(\beta_n, \ldots, \beta_0).
\end{equation}
Analogous constructions yield the left, right, and doubled Kirkwood–Dirac and Margenau–Hill temporal states. Their interrelations, as well as connections to pseudo-density operators, are discussed in detail in Ref.~\cite{Jia2025TemporalKirkwoodDirac}.
The temporal state also admits a recursive representation~\cite{Lie2025stateovertime,Jia2025TemporalKirkwoodDirac,Parzygnat2023pdo,fullwood2023quantum,Liu2025PDO} as given in Eq.~\eqref{eq:TemporalStateExp}. In this case, the temporal link product is defined for the operator product of Choi matrices and the initial state within their overlapping Hilbert space \cite{Jia2025TemporalKirkwoodDirac}.
Different expressions have their own advantages, Eq.~\eqref{eq:TSdual} plays a crucial role in analyzing the sample complexity of TST task. The recursive expression~\eqref{eq:TemporalStateExp} facilitates solving the quantum evolution map once the temporal state is reconstructed from experimental data.

\vspace{1em}
\emph{Quantum snapshotting of temporal quasiprobability distributions.} --- We now introduce a postprocessing protocol for reconstructing arbitrary TQDs from experimental data using quantum instruments.

Recall that a quantum instrument $\mathcal{I} = \{\mathcal{I}_{\alpha}\}$ is a collection of completely positive, trace-nonincreasing (CPTNI) maps such that $\sum_{\alpha} \mathcal{I}_{\alpha}$ is CPTP. For a density operator $\rho$, the outcome $\alpha$ occurs with probability $p(\alpha) = \Tr[\mathcal{I}_{\alpha}(\rho)]$, and the corresponding post-measurement state is $\rho' = \mathcal{I}_{\alpha}(\rho)/p(\alpha)$. Each map admits a Kraus representation,
\begin{equation}\label{eq:intrument}
  \mathcal{I}_{\alpha}(\rho) = \sum_{l} M_{\alpha,l}\,\rho\, M_{\alpha,l}^{\dagger},  
\end{equation}
with associated POVM elements $F_{\alpha,l} = M_{\alpha,l}^{\dagger} M_{\alpha,l}$ satisfying $\sum_{\alpha,l} F_{\alpha,l} = \mathbb I$.

The phase-space operation $\mathcal{P}^{t_0}_{q_0}$ is not, in general, a valid quantum operation, as it need not be completely positive or even Hermiticity preserving. To overcome this limitation, we decompose such operations into linear combinations of CPTNI maps, allowing the TQD to be reconstructed via classical postprocessing of measurement outcomes.

\begin{theorem}\label{thm:snapshotting}
There exists a quantum instrument $\{\mathcal{I}_{\alpha}\}$ that forms a basis for the superoperator space: each $\mathcal{I}_{\alpha}$ is a CPTNI map, $\sum_{\alpha} \mathcal{I}_{\alpha}$ is CPTP, and any linear superoperator $\mathcal{F}: \mathbf{B}(\mathcal{H}_{\tt i}) \to \mathbf{B}(\mathcal{H}_{\tt o})$ admits a decomposition
\begin{equation}
    \mathcal{F}(\rho) = \sum_{\alpha} \chi_{\alpha}\,\mathcal{I}_{\alpha}(\rho),
\end{equation}
where $\chi_{\alpha} \in \mathbb{C}$.
\end{theorem}

\begin{revision}
\begin{proof}
    For $\mathcal{H}_{\mathrm{i}} \otimes \mathcal{H}_{\mathrm{o}}$, choose an IC-POVM $\{F_{\alpha}\}$. Define
\begin{equation}
    \mathcal{I}_{\alpha}(\rho) = \frac{1}{d_{\mathrm{o}}}
    \operatorname{Tr}_{\mathrm{i}}
    \!\left[
    (\rho^{T} \otimes \mathbb I_{\mathrm{o}})
    F_{\alpha}
    \right],
\end{equation}
which is completely positive since $F_{\alpha}$ is positive.
It is clear that $\sum_{\alpha} \mathcal{I}_{\alpha}$ has Choi matrix
\begin{equation}
    \Phi_{\mathcal{I}} = \frac{1}{d_{\mathrm{o}}}\sum_{\alpha} F_{\alpha} = \frac{1}{d_{\mathrm{o}}} \mathbb I_{\mathrm{i}} \otimes \mathbb I_{\mathrm{o}}.
\end{equation}
We have $\operatorname{Tr}_{\mathrm{o}} \Phi_{\mathcal{I}} = \mathbb I_{\mathrm{i}}$, hence $\{\mathcal{I}_{\alpha}\}$ forms a quantum instrument.
Since $\{F_{\alpha}\}$ span $\mathbf{B}(\mathcal{H}_{\mathrm{i}} \otimes \mathcal{H}_{\mathrm{o}})$, $\{\mathcal{I}_{\alpha}\}$ span the space of linear superoperator from $\mathbf{B}(\mathcal{H}_{\mathrm{i}} )$ to $\mathbf{B}( \mathcal{H}_{\mathrm{o}})$. See Appendix for more details.
\end{proof}
\end{revision}

Theorem~1 of Ref.~\cite{Wang2024KD} treats the probe family associated with a selected projective basis. In contrast, our statement uses a single instrument basis that spans all linear maps, including those between unequal input and output dimensions. It is therefore an algebraic extension, rather than an efficiency improvement over that specialized construction.
Theorem~1 of Ref.~\cite{Wang2024KD} can be viewed as a special case of the above result. 
Note that Eq.~\eqref{eq:intrument} further implies that any $\mathcal{F}$ admits a decomposition as a linear combination of Kraus-type CPTNI maps $\mathcal{K}_{\alpha,l}(\bullet) := M_{\alpha,l} (\bullet)\, M_{\alpha,l}^{\dagger}$, with $\sum_{\alpha,l} M_{\alpha,l}^{\dagger} M_{\alpha,l} = \mathbb I$.
Refining the the outcome of quantum instrument into Kraus outcomes also yields a spanning collection of single-Kraus CPTNI maps, whose Kraus operators satisfy the usual completeness relation.

\begin{figure*}[t]
\centering
\includegraphics[width=0.96\textwidth]{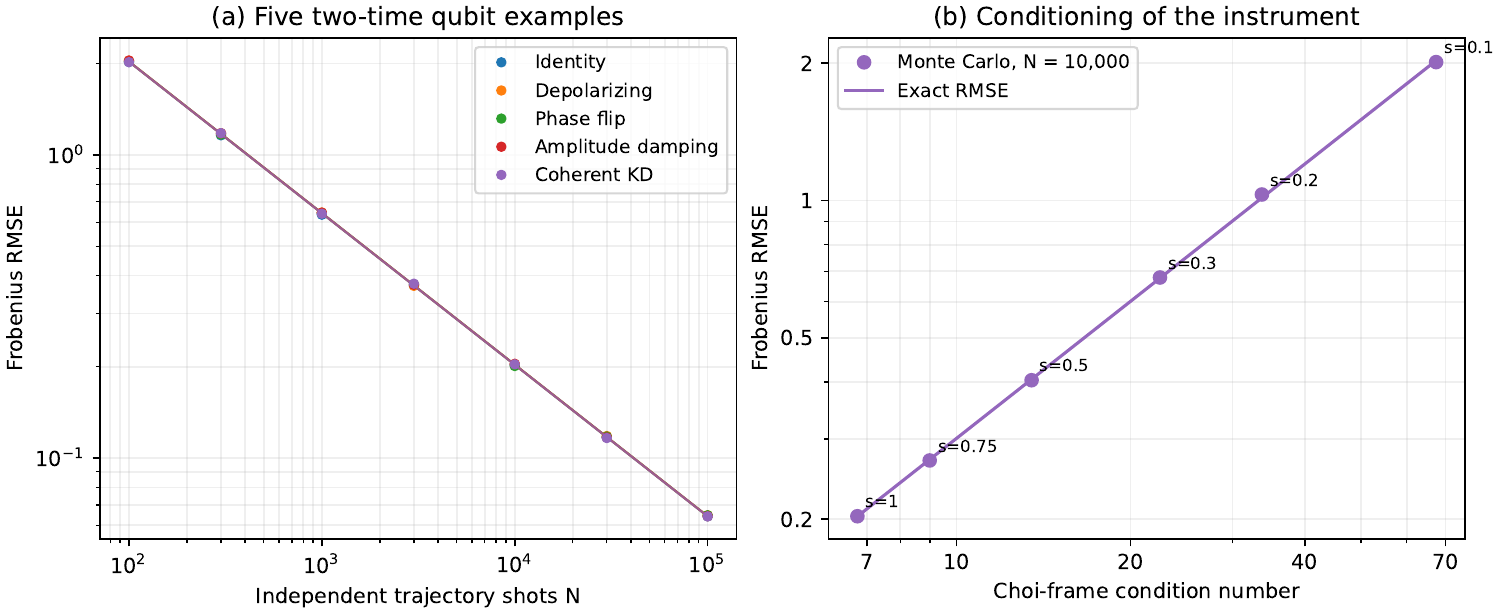}
\caption{Synthetic-shot two-time qubit TST using the normalized 16-branch instrument of Eq.~\eqref{eq:explicitbasis} at the$t_0$ and a four-outcome tetrahedral IC-POVM at $t_1$. (a) Markers: Frobenius root mean squared error over 1,000 independent multinomial experiments per shot count; lines: the exact $\sqrt{V(p)/N}$ prediction. (b) The complex KD example at $N=10^4$, with the first instrument attenuated by $s$ as in Eq.~\eqref{eq:attenuation}. Smaller $s$ worsens conditioning and amplifies statistical error. These are simulations with ideal calibrated instruments, not experimental data.}
\label{fig:qubits}
\end{figure*}

\begin{revision}
To give an explicit example of the decomposition, we fix a product basis of
$\mathcal{H}_{\mathrm{i}} \otimes \mathcal{H}_{\mathrm{o}}$, with
$d=d_{\mathrm{i}}d_{\mathrm{o}}$, and label the basis vectors $|kl\rangle$
($k=1,\dots,d_{\mathrm{i}}$, $l=1,\dots,d_{\mathrm{o}}$) in lexicographic order,
identifying $\alpha\equiv(k,l)\in\{1,\dots,d\}$ and writing
$|\alpha\rangle\equiv|kl\rangle$ when convenient. We define the following positive
operator basis on $\mathcal{H}_{\mathrm{i}} \otimes \mathcal{H}_{\mathrm{o}}$.
We first have the $d$ diagonal terms, one for each lexicographic index
$\alpha=1,\dots,d$,
\begin{equation}
    \Pi_{\alpha\alpha} = |\alpha\rangle \langle \alpha|.
\end{equation}
For every pair $\alpha<\beta$ in $\{1,\dots,d\}$, we include the following
\begin{align}
  \Pi_{\alpha\beta}^{(+)}
&= \frac{1}{2}\bigl(|\alpha\rangle + |\beta\rangle\bigr)\bigl(\langle \alpha| + \langle \beta|\bigr),\\
\Pi_{\alpha\beta}^{(-)}
&= \frac{1}{2}\bigl(|\alpha\rangle + i|\beta\rangle\bigr)\bigl(\langle \alpha| - i\langle \beta|\bigr).
\end{align}
There are $d + 2\binom{d}{2} = d^2$ such operators.
We rename these operators as $\{P_{\alpha}:\alpha=1,\cdots, d^2\}$, they are linearly independent and positive
semidefinite, hence form a positive operator basis on
$\mathcal{H}_{\mathrm{i}} \otimes \mathcal{H}_{\mathrm{o}}$. Define
\begin{equation}
  P := \sum_{\alpha=1}^{d^2} P_\alpha,
\qquad
K_\alpha := P^{-1/2}P_\alpha P^{-1/2},
\end{equation}
so that $K_\alpha\ge 0$ and $\sum_{\alpha=1}^{d^2} K_\alpha = \mathbb I$, which forms an
informationally complete POVM.
The associated CPTNI maps are obtained via the Choi isomorphism,
\begin{equation}
\mathcal{I}_\alpha(\rho)
=
\frac{1}{d_{\rm o}}
\operatorname{Tr}_{\mathrm{i}}\!\left[(\rho^{T}\otimes \mathbb I_{\mathrm{o}})\,K_\alpha\right].  \label{eq:explicitbasis}
\end{equation}
This set $\{\mathcal{I}_\alpha\}$ forms a basis of the space of superoperators.
Define the Gram matrix
\begin{equation}
    G_{\alpha\beta} = \operatorname{Tr}(K_{\alpha}K_{\beta}).
\end{equation}
For a superoperator $\mathcal{F}$, let its Choi matrix be
\(
\Phi_{\mathcal{F}}
\in \mathbf{B}(\mathcal{H}_{\mathrm{i}} \otimes \mathcal{H}_{\mathrm{o}}).
\)
We define the coefficients
\begin{equation}
    \chi_{\alpha} = d_{\rm o} \sum_{\beta} (G^{-1})_{\alpha\beta}\, \operatorname{Tr}\!\left(K_{\beta}\Phi_{\mathcal{F}}\right).
\end{equation}
This yields the decomposition of $\mathcal{F}$ in Theorem~\ref{thm:snapshotting}.
\end{revision}

Based on the above result, any phase-space operation can be realized by postprocessing the outcomes of the quantum instrument as
\begin{equation}
    \mathcal{P}_{q_k} = \sum_{\alpha} \chi_{\alpha} \mathcal{E}_{\alpha},
    \label{eq:phase_space_operation_CP_combination}
\end{equation}
through which we can obtain the TQD. We refer to this approach as \emph{quantum snapshotting}, following Ref.~\cite{Wang2024KD}.
Notice that for the left and right KD temporal states, at the final step, we need not implement the quantum snapshotting decomposition; instead, we can implement a IC-POVM at the final step, which is already a physical operation (since $\Tr (\bullet)F_{\alpha_n}= \Tr E_{\alpha_n}(\bullet)E_{\alpha_n}^{\dagger}$ with $F_{\alpha_n}=E_{\alpha_n}^{\dagger}E_{\alpha_n}$). For example, for the right KD temporal state, we have
\begin{equation}
\begin{aligned}
            &p(\alpha_n,\cdots,\alpha_0)\\
            = &\Tr  [ \mathcal{E}_{t_n\leftarrow t_{n-1}}\circ \cdots \circ \mathcal{I}_{\alpha_1}\circ \mathcal{E}_{t_1\leftarrow t_0}\circ\mathcal{I}_{\alpha_0}(\rho_{t_0})F_{\alpha_n}]
\end{aligned}
\end{equation}
Any right KD distribution can be written as a complex linear combination of this:
\begin{equation}
    \QR(\beta_n,\cdots,\beta_0)=\sum_{\alpha_n,\cdots,\alpha_0}\chi_{\beta_n\alpha_n}\cdots\chi_{\beta_0\alpha_0}p(\alpha_n,\cdots,\alpha_0).
\end{equation}
For simplicity, we can denote the postprocessing matrix as  $M_{\vec{\beta},\vec{\alpha}}$ and write $\overrightarrow{Q}_{\rm KD}(\vec{\beta})=\sum_{\vec{\alpha}} M_{\vec{\beta},\vec{\alpha}} p(\vec{\alpha})$.
However, for the doubled KD distribution, at the final step, we still need to implement the quantum snapshotting decomposition, since $F_{\alpha_n}(\bullet)F_{\alpha'_n}$ is generally not a physical operation.

\vspace{1em}
\emph{Temporal state tomography based on quantum snapshotting of temporal quasiprobability distributions.} ---
Given a multi-time quantum process, an IC-TQD fully determines the corresponding temporal state. By experimentally obtaining the TQD $Q(\beta_n,\cdots,\beta_0)$ via quantum snapshotting, one can construct an estimate $\Upsilon^e$ using the temporal Bloch representation~\cite{Jia2025TemporalKirkwoodDirac}. 
The temporal state tomography (TST) task is then formulated as the following optimization problem:
\begin{equation}\label{eq:Upse}
  \hat{\Upsilon} = \operatorname*{argmin}_{\Upsilon \in \mathbf{TS}(\mathcal{H}_{t_n}\otimes\cdots\otimes \mathcal{H}_{t_0})} \|\Upsilon - \Upsilon^e\|,
\end{equation}
where the norm can be selected for convenience; common choices include the trace norm and other Schatten $p$-norms (infidelity is not applicable here, as temporal states are generally non-positive semidefinite.)

As the IC-TQD is equivalent to the temporal state, a natural question is whether one can directly obtain an optimized TQD from experimental data $Q^e(\beta)$ via
\begin{equation}
  \hat{Q}(\vec{\beta}) = \operatorname*{argmin}_{Q \in \mathbf{TQD}} \|Q - Q^e\|.
\end{equation}
This procedure is not viable, however, since the full space of valid TQDs is typically difficult to characterize. We thus perform temporal state tomography entirely within the temporal state formalism.

\begin{table*}[t]
\caption{Worst-case comparison of independent parameters.}
\label{tab:comparison_parameters}

\footnotesize
\setlength{\tabcolsep}{3pt}
\renewcommand{\arraystretch}{1.15}

\begin{adjustbox}{max width=\textwidth}
\begin{tabular}{p{0.19\textwidth}p{0.29\textwidth}p{0.45\textwidth}}
\hline
Target
&
Independent parameters
&
Specified settings and outcomes
\\
\hline

Quantum state tomography (QST) on $\mathbb C^q$
&
$q^2-1$
&
One $q^2$-outcome IC-POVM.
\\

Quantum channel/process Tomography (QPT) $q\to q$
&
$q^4-q^2$
&
$q^2$ independent input preparations, each followed
by a $q^2$-outcome IC-POVM.
\\

Markov QST+QPT
&
$p_{\rm M}=(q^2-1)+n(q^4-q^2)$
&
QST and $n$ independently accessible QPT experiments.
\\

$n$-slot process tensor
&
$(q^2-1)\sum_{j=0}^{n}q^{4j}$;
Choi dimension $q^{2n+1}$
&
$q^4$ spanning CP intervention elements per slot
and $q^2$ final effects: $q^{4n+2}$ probability
coordinates. Measure-and-prepare realization:
$q^{2n}$ preparation strings and $q^{2(n+1)}$
outcome strings.
\\

Pseudo-density operator of $m$ qubit events
&
$4^m-1$ real Pauli correlators 
&
At most $4^m-1$ nonidentity sequential
Pauli-correlator settings. Identity at a time
means omitting that measurement.
\\


$m$-events TST: single local space
&
$q^{2m}-1$ generalized Pauli correlators
&
Universal construction: $q^{4m}$ instrument
branches.
\\

$m$-events TST: doubled local space
&
$q^{4m}-1$ generalized Pauli correlators

&
Universal construction: $q^{4m}$ instrument
branches.
\\

\hline
\end{tabular}
\end{adjustbox}
\end{table*}

This framework unifies state and process tomography: once $\Upsilon$ is reconstructed, reduced temporal marginals yield the states at individual time instances, while conditional temporal slices encode the dynamical map between successive times which can be reconstructed using Eq.~\eqref{eq:TemporalStateExp}. 
Hence, both static states and quantum evolutions are recovered within a single operational object.
For example, two-time state is given by $\Upsilon=\mathcal{E}_{t_1\leftarrow t_0}\star_{\mathrm{TS}}  \rho_{t_0}$, from TST obtained $\hat{\Upsilon}_{t_1t_0}$, we first calculate $\hat{\rho}_{t_0}$ as reduced temporal state and then solve the equation $\hat{\Upsilon}_{t_1t_0}= \hat{\mathcal{E}}_{t_1\leftarrow t_0}\star_{\mathrm{TS}}  \rho_{t_0}$ to obtain $\hat{\mathcal{E}}_{t_1\leftarrow t_0}$. The generalization to multi-time case is straightforward.

Equation~\eqref{eq:TemporalStateExp} also suggests a practical parameterization for the optimization: one may represent $\Upsilon$ in terms of the initial state and the Choi matrices of the intermediate channels, and optimize over these variables. Since both the initial state space and the Choi matrix space are well characterized, this renders the optimization tractable.

Different choices of temporal quasiprobability representations lead to distinct temporal state reconstructions \cite{Jia2025TemporalKirkwoodDirac,jia2026unifiedspatiotemporalquantumstates}:
\begin{itemize}
\item \emph{Left/right Kirkwood--Dirac (KD) TQD:} These yield the left/right KD temporal states.
\item \emph{Doubled Kirkwood--Dirac (KD) TQD:} This yields the doubled density operator; in this case, the local Hilbert space is doubled as $\mathcal{H}_{t_i}=\mathcal{H}_{t_i,L}\otimes \mathcal{H}_{t_i,R}$.
\item \emph{Left/right Margenau--Hill (MH) TQD:} Defined as the real part of the left/right KD TQD, these coincide and yield the left/right MH temporal states, and the resulting MH temporal state can be regarded as the Hermitian part of the left/right KD temporal states. In the two-time case, the MH temporal state reduces to the pseudo-density operator.
\item \emph{Doubled Margenau--Hill (MH) TQD:} Defined as the real part of the doubled KD TQD, this yields the doubled MH temporal state. The local Hilbert space is likewise doubled, and the resulting MH temporal state can be regarded as the Hermitian part of the doubled density operator.
\end{itemize}
We emphasize that the doubled temporal state possesses a local dimension of $d_t^2$, where $d_t$ denotes the dimension of the $t$-th Hilbert space. This doubling of the local dimension increases the complexity of the corresponding TST task.

\begin{revision}
When implementing TST, we also need to specify which type of temporal state we are dealing with, e.g., left KD, right KD, doubled KD, or MH states~\cite{Jia2025TemporalKirkwoodDirac,jia2026unifiedspatiotemporalquantumstates}. The set of such states can be characterized explicitly, so that the optimization in Eq.~\eqref{eq:Upse} can be taken over it.
As an example, consider the three-time right KD temporal state set
\begin{equation}
\begin{aligned}
     \mathbf{TS}=\{&\overrightarrow{\Upsilon}_{t_2,t_1,t_0}=J_2 J_1 \rho_{t_0}:\rho_{t_0}\geq 0, \Tr \rho_{t_0}=1,\\
     &J_i\geq 0, \Tr_{t_{i}} J_{i} =I_{t_{i-1}}\}.
\end{aligned}
\end{equation}
For other types of temporal state, an explicit expression can likewise be given based on the recursive expression.
\end{revision}

\vspace{1em}
\emph{Estimating sample complexity.} ---
Treating the temporal quantum process as a black box, one may ask for the minimal number of copies required for faithful reconstruction of the temporal state, which defines its sample complexity.

Since the temporal state $\Upsilon$ admits a linear reconstruction from the TQD and shares the same operator structure as a density operator on a Hilbert space of dimension
$d = \prod_{i=0}^n d_{t_i}$,
the estimation problem reduces to that of an informationally complete tomography task in a $d$-dimensional operator space. We assume that at each time step there are $m_i$ quantum instrument maps in the decomposition underlying the quantum snapshotting scheme, and denote $M=\prod_{i=0}^n m_i$.

Let $\hat{\Upsilon}$ be constructed from experimental estimates of the TQD. Standard concentration results for informationally complete measurements imply that $\hat{\Upsilon}$ concentrates around $\Upsilon$, with the estimation error governed by classical statistical fluctuations of the underlying quasiprobability distribution. The temporal state $\Upsilon$ admits the expansion
\begin{equation}
\Upsilon = \sum_{\vec{\beta}} G_{\vec{\beta}}\, Q(\vec{\beta}), \qquad
Q(\vec{\beta}) = \sum_{\vec{\alpha}} M_{\vec{\beta},\vec{\alpha}}\, p(\vec{\alpha}).
\end{equation}
where $Q(\vec{\beta})$ is TQD and $p(\vec{\alpha})$ the probability distribution based on time-ordered quantum instrument, $M_{\vec{\beta},\vec{\alpha}}$ is the postprocessing matrix.
Define
\begin{equation}
T_{\vec{\alpha}} := \sum_{\vec{\beta}} G_{\vec{\beta}}\, M_{\vec{\beta},\vec{\alpha}}, 
\end{equation}
so that
\begin{equation}
\Upsilon = \sum_{\vec{\alpha}} T_{\vec{\alpha}}\, p(\vec{\alpha}), \qquad
\hat{\Upsilon} = \sum_{\vec{\alpha}} T_{\vec{\alpha}}\, \hat{p}(\vec{\alpha}).
\end{equation}
Introduce the linear map
\begin{equation}
T: \mathbb{R}^M \to \mathbf{B}(\mathcal{H}_{t_n}\otimes \cdots \otimes \mathcal{H}_{t_0}), \quad
T(x) = \sum_{\vec{\alpha}} T_{\vec{\alpha}} x_{\vec{\alpha}}.
\end{equation}
and set ($p\in [1,\infty]$)
\begin{equation}
    \kappa_p=\|T\|_{2\to p}=\sup_{x\in\mathbb R^M,\,\|x\|_2=1}\left\|\sum_a x_{\vec{\alpha}}T_{\vec{\alpha}}\right\|_{p}.
\end{equation}
The above representation of temporal state directly relates the estimation error of the experimentally accessible probability distribution $p(\vec{\alpha})$ to that of the reconstructed temporal state.

Consequently, techniques from quantum state tomography extend naturally to the temporal setting, allowing one to bound the sample complexity of TST using techniques for informationally complete tomography.

\begin{theorem}[Sample complexity of TST]
Let $\Upsilon \in \mathbf{B}(\mathcal{H}_{t_n}\otimes \cdots \otimes \mathcal{H}_{t_0})$ be an $(n+1)$-step temporal quantum state, where $\dim(\mathcal{H}_{t_i}) = d_{t_i}$ and let $d = \prod_{i=0}^n d_{t_i}$ be total dimension and  $M=\prod_{i=0}^n m_i$ be total number of quantum instrument maps in quantum snapshotting. Then there exists a tomography scheme such that for any $\varepsilon,\delta \in (0,1)$,
\begin{equation}
\mathbb{P}\!\left(\|\hat{\Upsilon}-\Upsilon\|_p \ge \varepsilon\right) \le \delta
\end{equation}
using
\begin{equation}
N = O\!\left(\frac{\kappa_p^2 M}{\varepsilon^2}\log\frac{M}{\delta}\right)
\end{equation}
samples. 

\end{theorem}

Note that for the left and right KD temporal states, inserting the IC-POVM under the trace at the final step is a physical operation, so no quantum instrument decomposition is necessary and the complexity can be further reduced. For the doubled KD temporal state, however, we still need to decompose $F_{\alpha_n}(\bullet)F_{\alpha'_n}$. For simplicity, we ignore this difference and use the quantum instrument decomposition for all steps.

The sample complexity bound for temporal state tomography can be proven via nearly identical reasoning to its spatial counterpart. The key distinction is that a temporal state is not a positive semidefinite operator—a limitation resolved by the quantum snapshotting framework of TQD. Detailed derivations are given in the Appendix.

\vspace{1em}
\emph{Examples.} --- Let's first take two-time right KD temporal state as an example:
\begin{equation}\label{eq:qubitgeneral}
  \overrightarrow{\Upsilon}_{t_1t_0} =\frac{1}{4}\sum_{\mu,\nu=0}^3
 T^{\mu\nu}\sigma_\mu\otimes\sigma_\nu,
\end{equation}
where $T^{\mu\nu}=\operatorname{Tr}[\mathcal E(\rho\sigma_\nu)\sigma_\mu]$.
Below $X,Y,Z$ are Pauli matrices, $I=\mathbb I$ for identity, and $XX$ means $X\otimes X$, with analogous notation for other products.
We will implement simulation for the following cases:

(i) For $\rho_{t_0}=I/2$, the identity channel gives
\begin{equation}
  \overrightarrow{\Upsilon}_{t_1t_0}^{\rm id}
  =\tfrac14(II+XX+YY+ZZ).
\end{equation}
Its eigenvalues are $(1/2,1/2,1/2,-1/2)$.

(ii) Complete depolarization,
$\mathcal{D}(A)=\operatorname{Tr}(A)I/2$, instead gives
$\overrightarrow{\Upsilon}_{t_1t_0}^{\rm dep}=II/4$.
Both have marginals $\rho_{t_0}=\rho_{t_1}=I/2$, but
$T^{11}=T^{22}=T^{33}=1$ for the identity channel and zero for complete
depolarization. Thus, even both single-time density matrices together miss
the dynamical distinction. This example also explains why projecting the
temporal operator onto positive density matrices would discard valid
information.

(iii) Phase-flip noise
$\mathcal E_p(A)=(1-p)A+pZAZ$, at $\rho_{t_0}=I/2$, gives
\begin{equation}
 \overrightarrow{\Upsilon}_{t_1t_0}^{(p)}
 =\tfrac14[II+(1-2p)(XX+YY)+ZZ].
\end{equation}
The marginals again stay maximally mixed. The loss of transverse
correlations, $T^{11}=T^{22}=1-2p$, identifies the dephasing strength; we will
use $p=0.3$ in the simulations.

\begin{table*}[t]
\caption{Worst-case comparison of sample complexity.
$\varepsilon_{\rm F}$ and $\varepsilon_1$ denote Frobenius and trace-norm accuracy. The bounds containing $\delta$ guarantee
success $1-\delta$; the others use fixed constant success probability.}
\label{tab:comparison_samples}

\footnotesize
\setlength{\tabcolsep}{4pt}
\renewcommand{\arraystretch}{1.15}

\begin{adjustbox}{max width=\textwidth}
\begin{tabular}{p{0.38\textwidth}p{0.55\textwidth}}
\hline
Target
&
Sample complexity
\\
\hline

Full-rank $q$-dimensional state
&
$O(q^2/\varepsilon_1^2)$; $O(q/\varepsilon_{\rm F}^2)$
\cite{ODonnell2016QST,Haah2017QST}
\\


Full-rank $q\to q$ channel;
MUB projected least squares
&
$O(q^6\log(q^2/\delta)/\varepsilon_1^2)$; $O(q^4\log(q^2/\delta)/\varepsilon_{\rm F}^2)$ \cite{SurawyStepney2022projectedleast}
\\

Full-rank $q\to q$ channel;
diamond error $\varepsilon_\diamond$
&
$\Theta(q^4/\varepsilon_\diamond^2)$ with coherent
query processing;
$\Theta(q^6/\varepsilon_\diamond^2)$ without
quantum memory between queries \cite{bravoprieto2026quantummemoryadvantagequantum,Mele2026Optimal}
\\

Temporal state; quantum snapshotting 
&
Theorem~2
\\
Process tensor;
IC tester with linear inversion
&
$O(\kappa_{\rm F}^2\log(1/\delta)/\varepsilon_{\rm F}^2)$;
$O(\kappa_1^2\log(1/\delta)/\varepsilon_1^2)$ \cite{White2022processtom,Taranto2026HigherorderQ}
\\

$m$-event qubit pseudo-density operator;
equal allocation over Pauli correlators
&
$O(D^3\log(1/\delta)/\varepsilon_{\rm F}^2)$;
$O(D^4\log(1/\delta)/\varepsilon_1^2)$,
where $D=2^m$.
\\

\hline
\end{tabular}
\end{adjustbox}
\end{table*}

(iv) Amplitude damping with
$K_0=|0\rangle\langle0|+\sqrt{1-\gamma}|1\rangle\langle1|$ and
$K_1=\sqrt\gamma|0\rangle\langle1|$, at $\rho_{t_0}=I/2$, gives
\begin{equation}\begin{split}
 \overrightarrow{\Upsilon}_{t_1t_0}^{(\gamma)}
 =\tfrac14[II+\gamma ZI
 +\sqrt{1-\gamma}(XX+YY)\\+(1-\gamma)ZZ].
\end{split}\end{equation}
Its $t_1$ marginal is $(I+\gamma Z)/2$ and $t_0$ marginal is $I/2$.
For this assumed one-parameter model, the output alone already identifies
$\gamma$; TST additionally verifies the predicted longitudinal and
transverse correlations,
$T^{33}=1-\gamma$ and
$T^{11}=T^{22}=\sqrt{1-\gamma}$.
We will use $\gamma=0.4$ in simulation.

(v) With identity dynamics and $\rho=(I+rZ)/2$,
\begin{equation}\begin{split}
 \overrightarrow{\Upsilon}_{t_1t_0}^{(r)}
 =\tfrac14[II+XX+YY+ZZ\\
 +r(IZ+ZI+iYX-iXY)].
\end{split}\end{equation}
Here $T^{21}=ir$ and $T^{12}=-ir$: the right KD temporal state is
non-Hermitian for $r\neq0$. Its MH counterpart is
$(\overrightarrow{\Upsilon}_{t_1t_0}^{(r)}
+\overrightarrow{\Upsilon}_{t_1t_0}^{(r)\dagger})/2$. We will call this case the coherent KD in our simulation.
We simulate $r=0.6$ to test reconstruction with complex weights, as well
as the four Hermitian cases above.

For all five cases we reconstruct from the same $M=64$ trajectories. Figure~\ref{fig:qubits} compares the empirical Frobenius root mean squared error (RMSE) with the exact variance formula. To vary conditioning, let $J_\alpha=J[\mathcal{I}_{\alpha}]$ be the Choi matrices for quantum instrument at $t_0$ and set
\begin{equation}\label{eq:attenuation}
 J_\alpha(s)=sJ_\alpha+(1-s)c_\alpha I_4/2,\quad
 c_\alpha=\operatorname{Tr}(J_\alpha)/2.
\end{equation}
For $0<s\leq1$ this is still a normalized IC instrument; at $s=0$ all branches become proportional and completeness is lost.

We also simulate the TST for three-time right KD state and two-time doubled KD state. See Appendix for more details about the simulation.

\begin{revision}

\vspace{1em}
\emph{Conclusion and discussion.}--
In this work, we introduce TST as a new quantum tomography task that unifies the reconstruction of states and dynamical processes within a single multi-time object. Beyond this formal unification, the framework provides information that is inaccessible from a collection of single-time density operators alone. For example, for a maximally mixed qubit input, the identity channel and the completely depolarizing channel produce identical initial and final density operators, while their right KD temporal states are different and therefore retain the dynamical distinction. More generally, temporal correlations are encoded directly in the reconstructed operator and need not be inferred from separate tomographic reconstructions performed at individual times.
A practical feature of the construction is that the same informationally complete instrument can be used throughout the experiment, while different phase-space probes are synthesized entirely through classical post-processing. Thus, changing the quasiprobability representation or the reconstructed temporal observable does not require a corresponding modification of the measurement apparatus. In particular, the protocol does not rely on weak measurements or interferometric implementations of quasiprobabilities: the required information is extracted from ordinary instrument outcome statistics and subsequently combined linearly. This is especially useful for temporal-state constructions, which are generally nonpositive or even non-Hermitian and therefore cannot be treated directly using standard density-matrix tomography.

We emphasize that states, channels, and multi-time processes already have a common description in the quantum-comb and process-tensor frameworks~\cite{Chiribella2009comb,Pollock2018processtensor}; process matrices additionally allow indefinite causal order~\cite{oreshkov2012quantum}, and their local spaces are all doubled, similar to the doubled KD temporal state.
The temporal state, depends on the choice of phase space distributions can habe single local space or doubled local space.
It also has nice property that all equal time state are reduced states of the corresponding temporal states and all quantum channels can be obtained via recursive expression for the temporal state \cite{Jia2025TemporalKirkwoodDirac,jia2026unifiedspatiotemporalquantumstates}.

Tables~\ref{tab:comparison_parameters}--\ref{tab:comparison_samples} make a worst-case comparison, namely, no low-rank, sparsity, bounded-memory, or known-Markov-order promise is imposed. 
Note for these comparison, our framework use a universal construction: For each time step, for left, right and doubled IC-POVM inserting, we use the same setting of quantum instrument to implement the quantum snapshotting. For $m$ events, assume local dimension is $q$, we have $q^{4(m-1)}$ instrument branches at first $m-1$ slots and $q^2$ final IC-POVM outcomes.

Several directions merit further investigation. First, it would be useful to identify optimal informationally complete instrument bases for decomposing phase-space superoperators and to characterize their reconstruction conditioning systematically. A closely related problem is to obtain sharper sample-complexity bounds by adapting techniques from spatial quantum state tomography to the multi-time setting. It would also be interesting to determine whether nonclassical features of temporal quasiprobability distributions, such as negativity or complex-valuedness, are directly related to tomographic efficiency or statistical robustness. This question is closely connected to non-Markovianity and will be investigated in Ref.~\cite{Kelvin2025KD}. Finally, a systematic comparison between TST, process-tensor tomography, and quantum-comb tomography would clarify their respective operational regimes and resource requirements~\cite{White2022processtom,Antesberger2024highertomography,Li2025combtomography}. These questions are left for future work.

\end{revision}

\begin{acknowledgments}
Z. J. thanks Dagomir Kaszlikowski, Kavan Modi, Kelvin Onggadinata, Koh Teck Seng, John Kam, Yunlong Xiao, and Zhihao Ma for discussions on temporal quasiprobabilities.
Z.J. is supported by the National Natural Science Foundation of China (Grant
No. 12605015), Start-up Grant of Central South University (Grant No. 502045031) and by the Frontier Interdisciplinary Direction ``Quantum Technologies'' Project of Central South University (Grant No. 506010805).
\end{acknowledgments}

\appendix

\section{Proof of Theorem~\ref{thm:snapshotting}}
\label{app:KNnogo}

In this section, we provide a self-contained proof of Theorem~\ref{thm:snapshotting} from the main text.
We begin by recalling the Choi-Jamio\l kowski (CJ) isomorphism, a fundamental correspondence in quantum information theory that establishes a linear bijection between linear maps
$\mathcal F : \mathbf B(\mathcal H_{\mathrm{i}}) \to \mathbf B(\mathcal H_{\mathrm{o}})$
and operators acting on the tensor product Hilbert space
$\mathcal H_{\mathrm{i}} \otimes \mathcal H_{\mathrm{o}}$.

Fix an orthonormal basis $\{ |i\rangle \}$ of $\mathcal H_{\mathrm{i}}$.
For any linear map $\mathcal F$, its Choi operator is defined as
\begin{equation}
\Phi_{\mathcal F}
=
\sum_{i,j}
|i\rangle\langle j|
\otimes
\mathcal F(|i\rangle\langle j|)
\;\in\;
\mathbf B(\mathcal H_{\mathrm{i}} \otimes \mathcal H_{\mathrm{o}}).
\end{equation}
Conversely, given an operator
$\Phi_{\mathcal F} \in \mathbf B(\mathcal H_{\mathrm{i}} \otimes \mathcal H_{\mathrm{o}})$,
the corresponding linear map $\mathcal F$ is recovered via
\begin{equation}
\mathcal F(\rho)
=
\operatorname{Tr}_{\mathrm{i}}
\!\left[
(\rho^{T} \otimes \mathbb I_{\mathrm{o}})\,
\Phi_{\mathcal F}
\right],
\end{equation}
where the transpose is taken with respect to the chosen basis
$\{ |i\rangle \}$.

The Choi--Jamio\l kowski isomorphism thus establishes a linear bijection
between linear maps $\mathcal F$ and their Choi operators
$\Phi_{\mathcal F}$.
Moreover, Choi's theorem asserts that
$\mathcal F$ is completely positive if and only if
\begin{equation}
\Phi_{\mathcal F} \ge 0.
\end{equation}
Furthermore, $\mathcal F$ is trace-preserving if and only if
\begin{equation}
\operatorname{Tr}_{\mathrm{o}} \Phi_{\mathcal F}
=
\mathbb I_{\mathrm{i}} .
\end{equation}
This give a complete characterization of space of CPTP maps.

\begin{lemma}\label{lemma:Herm}
Let $\mathcal{H}$ be a (possibly infinite-dimensional) Hilbert space. Any bounded operator
$X \in \mathbf{B}(\mathcal{H})$ can be written as a finite linear combination of positive semidefinite operators.
\end{lemma}

\begin{proof}
Any operator $X \in \mathbf{B}(\mathcal{H})$ can be decomposed into linear combination of Hermitian operators
\begin{equation}
    X = \frac{X + X^{\dagger}}{2} + i\,\frac{X - X^{\dagger}}{2i},
\end{equation}
where both operators $(X + X^{\dagger})/2$ and $(X - X^{\dagger})/(2i)$ are Hermitian.
By the spectral decomposition theorem, any Hermitian operator $H$ admits a decomposition $H = H_{+} - H_{-}$, where $H_{+}$ and $H_{-}$ are positive semidefinite operators corresponding to the positive and negative parts of the spectrum. Consequently, each Hermitian operator is a real linear combination of positive semidefinite operators.
Combining these decompositions shows that $X$ itself can be expressed as a finite linear combination of positive semidefinite operators.
\end{proof}

\begin{lemma}\label{lemma:CP}
Any completely positive (CP) map $\mathcal F$ can be rescaled into a completely positive, trace non-increasing (CPTNI) map.
\end{lemma}

\begin{proof}
Let $\mathcal F$ be a CP map, and denote its Choi operator by $\Phi_{\mathcal F} \ge 0$.  
Define
\begin{equation}
R := \operatorname{Tr}_{\mathrm{o}} \Phi_{\mathcal F} \ge 0,
\end{equation}
and let $\lambda > 0$ be the spectral radius of $R$.  

We then define a new CP map $\mathcal F'$ whose Choi operator is
\begin{equation}
\Phi_{\mathcal F'} := \frac{\Phi_{\mathcal F}}{\lambda}.
\end{equation}
Its marginal satisfies
\begin{equation}
R' := \operatorname{Tr}_{\mathrm{o}} \Phi_{\mathcal F'} = \frac{R}{\lambda} \le \mathbb I_{\mathrm{i}}.
\end{equation}
Since $\Phi_{\mathcal F'} \ge 0$ and $R' \le \mathbb I_{\mathrm{i}}$, the map $\mathcal F'$ is CPTNI.
\end{proof}

\begin{theorem}\label{thm:CPTNI-decom}
Any linear superoperator $\mathcal{F}: \mathbf{B}(\mathcal{H}_{\mathrm{i}}) \to \mathbf{B}(\mathcal{H}_{\mathrm{o}})$ can be expressed as a complex linear combination of CPTNI maps,
\begin{equation}
    \mathcal{F}(\rho) = \sum_{\alpha} \chi_{\alpha} \mathcal{E}_{\alpha}(\rho),
    \label{eq:quantum_operation_CP_combination_F}
\end{equation}
where $\chi_{\alpha} \in \mathbb{C}$ denote the complex coefficients of the linear combination, and $\{\mathcal{E}_{\alpha}\}$ are CPTNI maps.
\end{theorem}

\begin{proof}
By combining Lemmas~\ref{lemma:Herm} and \ref{lemma:CP}, we obtain the desired conclusion.
\end{proof}

Notice that Theorem~\ref{thm:CPTNI-decom} is weaker than Theorem~\ref{thm:snapshotting}. 
Theorem~\ref{thm:snapshotting} can be proven using the following well-known result in quantum information theory:

\begin{lemma}
For any Hilbert space $\mathcal{H}$ with dimension $d$, there exists an informationally complete POVM (IC-POVM). That is, there exists a set of positive operators $\{F_k\}$ which forms a basis for the space of Hermitian operators $\Herm(\mathcal{H})$ as a real vector space.
\end{lemma}

\begin{proof}
Let $\{A_i\}$ be a set of positive semidefinite operators that span $\Herm(\mathcal{H})$, and define $A = \sum_i A_i$. Then
\begin{equation}
    F_i = A^{-1/2} A_i A^{-1/2}
\end{equation}
defines a POVM, since $\sum_i F_i = \mathbb I$.
To see that $\{F_i\}$ is informationally complete, note that $A = \sum_i A_i$ is full rank. Hence the map $X \mapsto A^{-1/2} X A^{-1/2}$ is an isomorphism on $\Herm(\mathcal{H})$, and therefore $\{F_i\}$ also spans $\Herm(\mathcal{H})$.
\end{proof}

\begin{corollary}
For any Hilbert space $\mathcal{H}$ with dimension $d$, there exists an IC-POVM which forms a basis for the space of linear operators $\mathbf{B}(\mathcal{H})$ as a complex vector space.
\end{corollary}

\begin{proof}
This follows from the fact that $\Herm(\mathcal{H})$ spans $\mathbf{B}(\mathcal{H})$ over $\mathbb{C}$.
\end{proof}

\section{Proof of sample complexity bound}

The sample complexity bound for temporal state tomography follows largely the same line of reasoning as its spatial counterpart. The key obstacle lies in the fact that a temporal state is not a positive semidefinite operator, a difficulty we resolve using the quantum snapshotting method of TQD.
Our proof relies on the Hoeffding inequality, stated below.

\begin{lemma}[Hoeffding inequality]
Let $X_1, \dots, X_n$ be independent random variables such that $X_i \in [a_i, b_i]$ almost surely, where $a_i \leq b_i$. Let $\mathbb{E}[X_i] = \mu_i$ for $1 \leq i \leq n$, and define $X = \sum_{i=1}^n X_i$ and $\mu = \mathbb{E}[X] = \sum_{i=1}^n \mu_i$. Then, for all $t \geq 0$,
\[
\mathbb{P}\big(|X - \mu| \geq t\big)
\leq 2 \exp\!\left(-\frac{2t^2}{\sum_{i=1}^n (b_i - a_i)^2}\right).
\]
\end{lemma}

\textbf{Proof of sample complexity bound in Theorem 2:}
Consider a time-ordered quantum instrument $\{\mathcal{I}_{\alpha_0},\dots,\mathcal{I}_{\alpha_n}\}$ and denote $\vec{\alpha}=(\alpha_0,\dots,\alpha_n)$. Let $p(\vec{\alpha})$ be the induced trajectory probability distribution.

Define the random variables (we use $\mathbf{1}$ to denote indicator function)
\begin{equation}
X_i^{\vec{\alpha}} := \mathbf{1}[\text{trajectory } \vec{\alpha} \text{ occurs in run } i],
\end{equation}
which are  independent and identically distributed (i.i.d.) and satisfy
\begin{equation}
\mathbb{E}[X_i^{\vec{\alpha}}] = p(\vec{\alpha}), \quad X_i^{\vec{\alpha}} \in [0,1].
\end{equation}
The empirical frequency is
\begin{equation}
\hat{p}(\vec{\alpha}) = \frac{1}{N}\sum_{i=1}^N X_i^{\vec{\alpha}}.
\end{equation}
By Hoeffding's inequality, for any fixed $\vec{\alpha}$,
\begin{equation}
\mathbb{P}\!\left(
\left| \hat{p}(\vec{\alpha}) - p(\vec{\alpha}) \right| \ge t
\right)
\le 2 \exp(-2 N t^2).
\end{equation}
Let $M := |\{\vec{\alpha}\}|$ be the number of trajectories. For informationally complete instruments, one has
\begin{equation}
M = \prod_{i=0}^n m_i .
\end{equation}
Applying a union bound over all $\vec{\alpha}$ yields
\begin{equation}
\mathbb{P}\!\left(
\max_{\vec{\alpha}} |\hat{p}(\vec{\alpha}) - p(\vec{\alpha})| \ge t
\right)
\le 2M \exp(-2Nt^2).
\end{equation}
Setting the right-hand side to be at most $\delta$, we obtain
\begin{equation}
t = \sqrt{\frac{1}{2N}\log\frac{2M}{\delta}}.
\end{equation}
Hence, with probability at least $1-\delta$,
\begin{equation}
\|\hat{p}-p\|_\infty \le \sqrt{\frac{1}{2N}\log\frac{2M}{\delta}}.
\end{equation}
Consequently,
\begin{equation}
\|\hat{p}-p\|_2
\le \sqrt{M}\,\|\hat{p}-p\|_\infty
\le \sqrt{\frac{M}{2N}\log\frac{2M}{\delta}}.
\end{equation}
Next, we use the linear reconstruction of the temporal state. By assumption, $\Upsilon$ admits the expansion
\begin{equation}
\Upsilon = \sum_{\vec{\beta}} G_{\vec{\beta}} Q(\vec{\beta}), \quad
Q(\vec{\beta}) = \sum_{\vec{\alpha}} M_{\vec{\beta},\vec{\alpha}} p(\vec{\alpha}).
\end{equation}
Define
\begin{equation}
T_{\vec{\alpha}} := \sum_{\vec{\beta}} G_{\vec{\beta}} M_{\vec{\beta},\vec{\alpha}},
\end{equation}
so that
\begin{equation}
\Upsilon = \sum_{\vec{\alpha}} T_{\vec{\alpha}} p(\vec{\alpha}), \quad
\hat{\Upsilon} = \sum_{\vec{\alpha}} T_{\vec{\alpha}} \hat{p}(\vec{\alpha}).
\end{equation}
Thus,
\begin{equation}
\hat{\Upsilon}-\Upsilon
= \sum_{\vec{\alpha}} T_{\vec{\alpha}} (\hat{p}-p)(\vec{\alpha}).
\end{equation}

Introduce the linear map
\begin{equation}
T: \mathbb{R}^M \to \mathbf{B}(\mathcal{H}_{t_n}\otimes \cdots \otimes \mathcal{H}_{t_0}), \quad
T(x) = \sum_{\vec{\alpha}} T_{\vec{\alpha}} x_{\vec{\alpha}}.
\end{equation}
Then
\begin{equation}
\|\hat{\Upsilon}-\Upsilon\|_2
\le \|T\|_2 \cdot \|\hat{p}-p\|_2.
\end{equation}
Set $\|T\|_2 = \kappa$, and therefore
\begin{equation}
\|\hat{\Upsilon}-\Upsilon\|_2
= O\!\left( \kappa \sqrt{\frac{M}{N}\log\frac{M}{\delta}}\right).
\end{equation}
Setting the right-hand side to be at most $\varepsilon$ gives
\begin{equation}
N = O\!\left(\frac{\kappa^2 M}{\varepsilon^2}\log\frac{M}{\delta}\right).
\end{equation}
We thus complete the proof.

\section{Qubit two-time TST simulation}
\label{app:simulation}

\begin{figure}[t]
\centering
\includegraphics[width=1\linewidth]{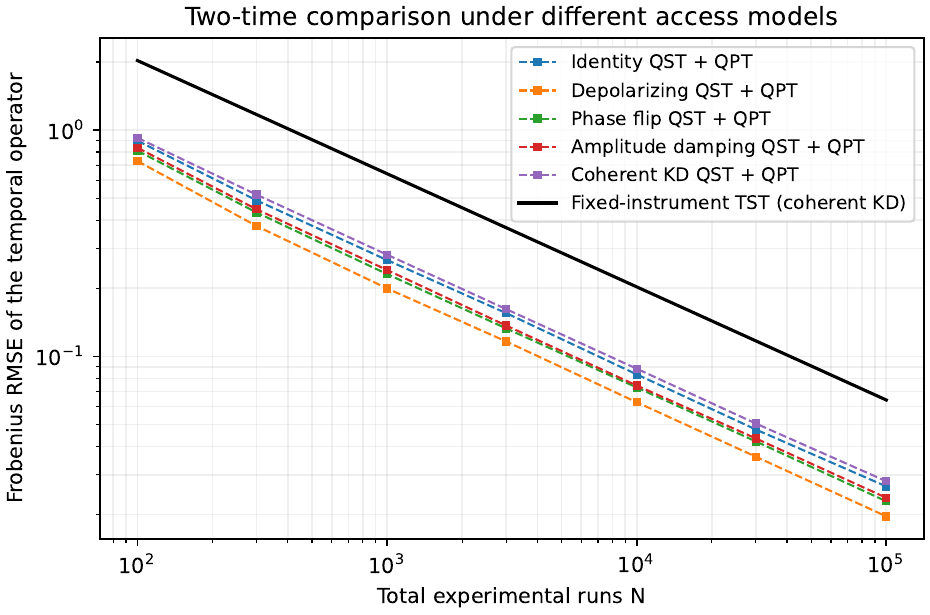}
\caption{Comparison with separate initial-state QST plus channel QPT at the same total number of experimental runs. Dashed curves show the empirical RMSEs obtained from 1000 independent trials for the five examples. The solid curve shows the exact TST RMSE for the coherent KD example; the other TST curves nearly coincide with it. Independent preparation of four channel inputs is available only to the QST--QPT baseline.}
\label{fig:baseline}
\end{figure}

For the two-time right KD temporal state at times $t_0$ and $t_1$, we use the tetrahedral IC-POVM and its dual frame,
\begin{equation}
 K_a=\tfrac14(I+\boldsymbol v_a\cdot\boldsymbol\sigma),\qquad
 G_a=6K_a-I,
\end{equation}
where $\boldsymbol v_a$ runs over $(1,1,1)/\sqrt3$, $(1,-1,-1)/\sqrt3$, $(-1,1,-1)/\sqrt3$, and $(-1,-1,1)/\sqrt3$. They satisfy
\begin{equation}
\operatorname{Tr}(G_aK_b)=\delta_{ab}.
\end{equation}
At the initial time $t_0$, the right-KD probe is generally not a physical quantum operation and is therefore synthesized through the quantum-snapshotting decomposition using the 16-outcome instrument of Eq.~\eqref{eq:explicitbasis}. At the final time $t_1$, no such decomposition is required, since the IC-POVM $\{K_b\}$ can be implemented directly.

The physical trajectory probabilities, KD quasiprobabilities, and corresponding matrix snapshots are
\begin{align}
 p_{\alpha b}
 &=
 \operatorname{Tr}\!\left[
 \mathcal E\!\left(\mathcal I_\alpha(\rho)\right)K_b
 \right],\\
 \overrightarrow{Q}_{\rm KD}(b,a)
 &=
 \sum_\alpha\chi_{a\alpha}p_{\alpha b},\\
 T_{\alpha b}
 &=
 \sum_{a=1}^4
 \chi_{a\alpha}G_b\otimes G_a.
\end{align}
They satisfy
\begin{equation}
\sum_{\alpha,b}p_{\alpha b}=1,
\qquad
\sum_{\alpha,b}
p_{\alpha b}T_{\alpha b}
=
\overrightarrow{\Upsilon}_{t_1,t_0}.
\end{equation}
Thus, although the physical trajectory probabilities contain the disturbance associated with the implemented instrument, the weighted classical postprocessing exactly recovers the target right-KD probe correlations.

Equation~\eqref{eq:attenuation} preserves positivity of each instrument branch. Moreover, since $c_\alpha\geq0$ and $\sum_\alpha c_\alpha=1$,
\begin{equation}
\sum_\alpha J_\alpha(s)
=
s\sum_\alpha J_\alpha
+
(1-s)\frac{I_4}{2},
\end{equation}
which is again the Choi matrix of a trace-preserving channel. Hence the attenuated branches form a normalized quantum instrument. On the Choi-operator space, the attenuation acts through the linear map
\begin{equation}
\mathcal A_s(A)
=
sA+(1-s)\operatorname{Tr}(A)\frac{I_4}{4},
\end{equation}
whose eigenvalue is $1$ on the identity subspace and $s$ on the traceless subspace. Therefore $\mathcal A_s$ is invertible for $s>0$, and informational completeness of the original Choi frame is preserved. The condition number shown in Fig.~\ref{fig:qubits}(b) is the ratio of the largest to the smallest singular value of the $16\times16$ matrix whose columns are $\operatorname{vec}J_\alpha(s)$; the condition number of the corresponding Gram matrix is its square. 

We use NumPy's default random number generator with seed $20260923$. For each sample size
\[
N\in\{100,300,10^3,3\times10^3,10^4,3\times10^4,10^5\},
\]
we generate 1000 independent simulated datasets. Each TST dataset consists of a single multinomial count vector with 64 categories and total count $N$. These categories correspond to the 16 outcomes of the first-time instrument combined with the four outcomes of the final-time IC-POVM; they are distinct from the 16 matrix entries of the two-time KD representation. Thus, $N$ denotes the total number of experimental runs rather than the number of runs per category. For the conditioning study, we fix $N=10^4$ and vary
\(
s\in\{1,0.75,0.5,0.3,0.2,0.1\}.
\)

Each experimental run produces a single matrix snapshot $T_{\alpha b}$. Hence the empirical TST estimator can equivalently be written as
\begin{equation}
\widehat{\Upsilon}_{t_1,t_0}
=
\frac{1}{N}
\sum_{i=1}^N T_{Y_i},
\end{equation}
where
\begin{equation}
\mathbb P\!\left(Y_i=(\alpha,b)\right)
=
p_{\alpha b}.
\end{equation}
Defining the single-shot snapshot variance by
\begin{equation}
V(p)
:=
\sum_{\alpha,b}
p_{\alpha b}
\left\|
T_{\alpha b}
-
\overrightarrow{\Upsilon}_{t_1,t_0}
\right\|_2^2,
\end{equation}
independence of the $N$ runs gives the exact identity
\begin{equation}
\mathbb E
\left\|
\widehat{\Upsilon}_{t_1,t_0}
-
\overrightarrow{\Upsilon}_{t_1,t_0}
\right\|_2^2
=
\frac{V(p)}{N}.
\end{equation}
Therefore the exact Frobenius RMSE is
\begin{equation}
\operatorname{RMSE}_{\rm TST}
=
\sqrt{\frac{V(p)}{N}}.
\end{equation}

For comparison, we also simulate a protocol combining initial-state quantum state tomography (QST) with quantum process tomography (QPT) of the channel. For each total budget $N$, we allocate $N/5$ runs to initial-state tomography and $N/5$ runs to each of four independently prepared channel inputs,
\[
\tau_a=2K_a,\qquad a=1,\ldots,4.
\]
The same tetrahedral IC-POVM is used to measure the initial state and each channel output. Let $\widehat{\rho}$ and $\widehat{\omega}_a$ denote the corresponding dual-frame estimates of $\rho$ and $\mathcal E(\tau_a)$, respectively. Since
\[
\operatorname{Tr}\!\left[
\frac{G_a}{2}\tau_b
\right]
=
\delta_{ab},
\]
the reconstructed channel acts as
\begin{equation}
\widehat{\mathcal E}(A)
=
\sum_{a=1}^{4}
\operatorname{Tr}\!\left[
\frac{G_a}{2}A
\right]
\widehat{\omega}_a.
\end{equation}
The corresponding temporal-state estimate is obtained by substituting $\widehat{\rho}$ and $\widehat{\mathcal E}$ into Eq.~\eqref{eq:qubitgeneral}. The initial-state and channel datasets are statistically independent. No positivity projection is applied to the state estimates, the reconstructed channel Choi matrix, or the resulting temporal operator.

The separate QST--QPT protocol uses five acquisition configurations and assumes access to four independently prepared channel inputs, whereas TST uses a single fixed sequential measurement protocol. All initial-state tomography runs are included in the QST--QPT budget. Consequently, the channel is used $4N/5$ times in the QST--QPT protocol and $N$ times in TST. Preparation and calibration costs are not included. The comparison therefore concerns these particular protocols under their respective experimental access assumptions and is not intended to establish minimax optimality or a general ordering of their sample complexities.

For the unattenuated TST instrument ($s=1$), the Choi-frame condition number is $6.7446$. The reconstruction norm is
\begin{equation}
\kappa_2
=
\sup_{x\in\mathbb R^{64},\,\|x\|_2=1}
\left\|
\sum_{\alpha,b}
x_{\alpha b}T_{\alpha b}
\right\|_2
=
54.9702,
\end{equation}
where the domain is real because the empirical trajectory probabilities and their fluctuations are real. Least-squares fits of the logarithm of the empirical RMSE against $\log N$ yield slopes ranging from $-0.5010$ to $-0.4990$ across the five examples, in agreement with the exact $N^{-1/2}$ dependence. Table~\ref{tab:qubit-rmse} lists the empirical RMSEs at $N=10^4$. For these five examples and this specific allocation of the total run budget, the separate QST--QPT protocol gives a lower empirical RMSE than TST.

\begin{table}[tb]
\caption{Empirical Frobenius RMSEs of the reconstructed
two-time temporal operator at a total budget of $N=10^4$
runs, evaluated over 1000 independent simulated datasets.
TST uses the unattenuated instrument ($s=1$).}
\label{tab:qubit-rmse}
\begin{ruledtabular}
\begin{tabular}{lcc}
Example & TST & QST + QPT \\
\hline
Identity          & 0.2026 & 0.0833 \\
Depolarizing      & 0.2038 & 0.0627 \\
Phase flip        & 0.2011 & 0.0729 \\
Amplitude damping & 0.2045 & 0.0744 \\
Coherent KD       & 0.2037 & 0.0883 \\
\end{tabular}
\end{ruledtabular}
\end{table}

For the coherent KD example at $N=10^4$, the exact TST RMSE increases from $0.2030$ at $s=1$ to $2.0269$ at $s=0.1$, while $\kappa_2$ increases from $54.97$ to $548.45$. Thus, the exact $N^{-1/2}$ scaling for a fixed instrument can coexist with a substantially larger error prefactor as the instrument becomes poorly conditioned. The five TST error curves at $s=1$ nearly coincide because their corresponding single-shot snapshot variances $V(p)$ are similar.

\begin{revision}
\begin{figure}
    \centering
\includegraphics[width=\linewidth]{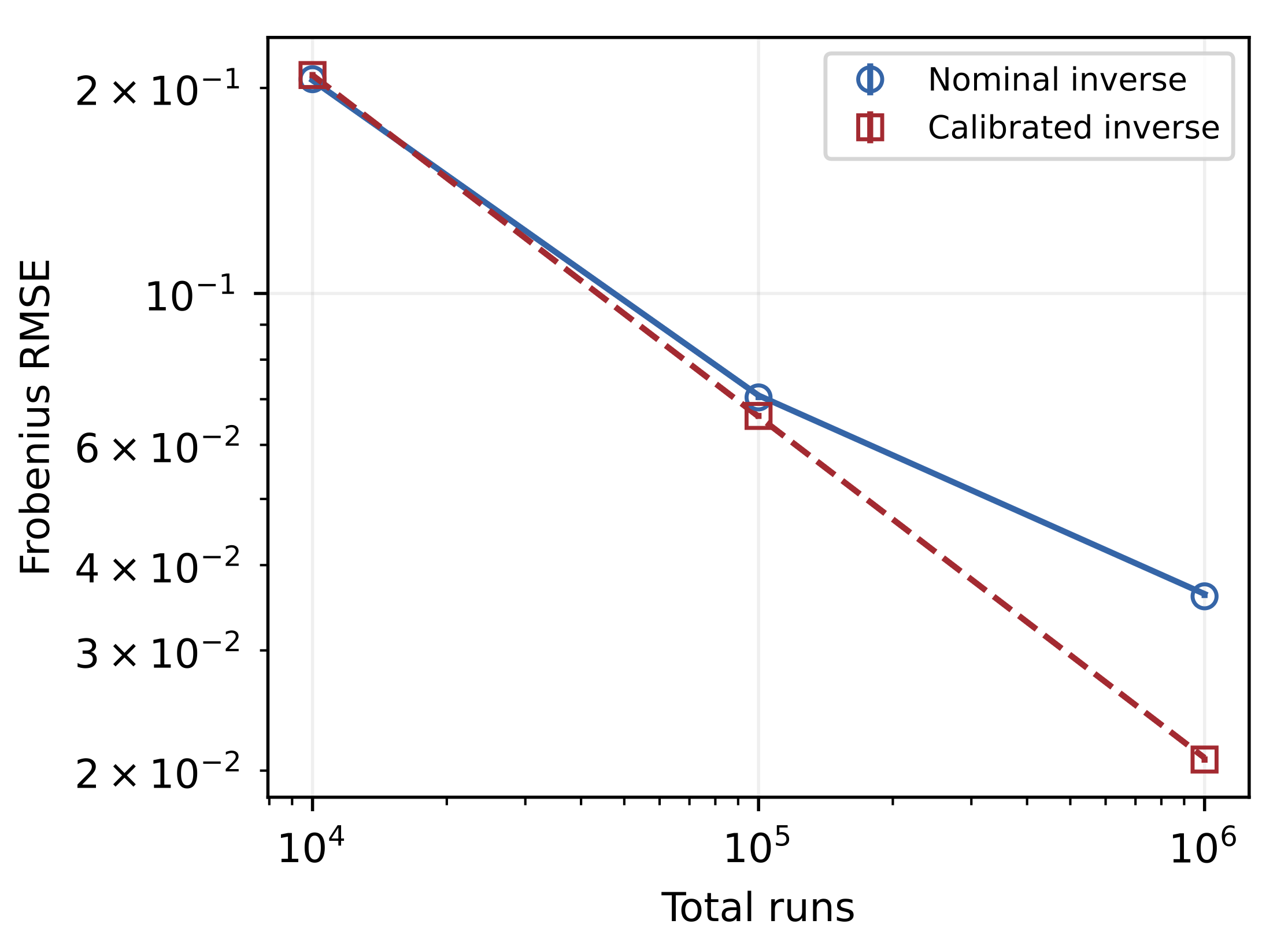}
    \caption{Two-time SPAM test at $u=1$, evaluated against the actual prepared target. Markers show empirical RMSE; error bars show one Monte Carlo standard error. Lines give the exact bias-plus-variance prediction.}
    \label{fig:twotimeCali}
\end{figure}

\subsection{Calibration errors}
For a two-time process with identity channel, we introduce a scale $u\in\{0,0.5,1,2\}$ and prepare $\rho_u=(1-0.02u)\rho+0.02u I/2$.
The implemented branch frame is $J_\alpha(s)$ with $s=1-0.01u$, and the final effects are
$K_b^{(u)}=(1-0.02u)K_b+0.02u I/4$.
All remain physical and informationally complete. The target is the right KD temporal state of the actually prepared $\rho_u$ and identity dynamics. Preparation error is thus not incorrectly labeled an inverse-calibration bias. Its distance from the intended target is separately nonzero ($0.01200$ at $u=1$).
We compare the nominal inverse with the exactly calibrated inverse of the implemented branch maps and final effects, the result is shown in Figure~\ref{fig:twotimeCali}. 
\end{revision}

\begin{revision}
\section{Three-time Kirkwood-Dirac TST and doubled Kirkwood-Dirac TST simulation}\label{app:newnumerics}

\subsection{Three-time Kirkwood-Dirac TST}
In this part, we simulate the three-time TST for qubits.
We consider the tetrahedral IC-POVM for each time step and use the 16-outcome quantum instrument at each time step of the explicit construction.  And the final step, the multiplication of IC-POVM is itself a physical operation, we thus do not need to decompose it into quantum instruments. Thus the record contains $16^2\times4=1024$ trajectories.

\begin{figure}
    \centering
    \includegraphics[width=\linewidth]{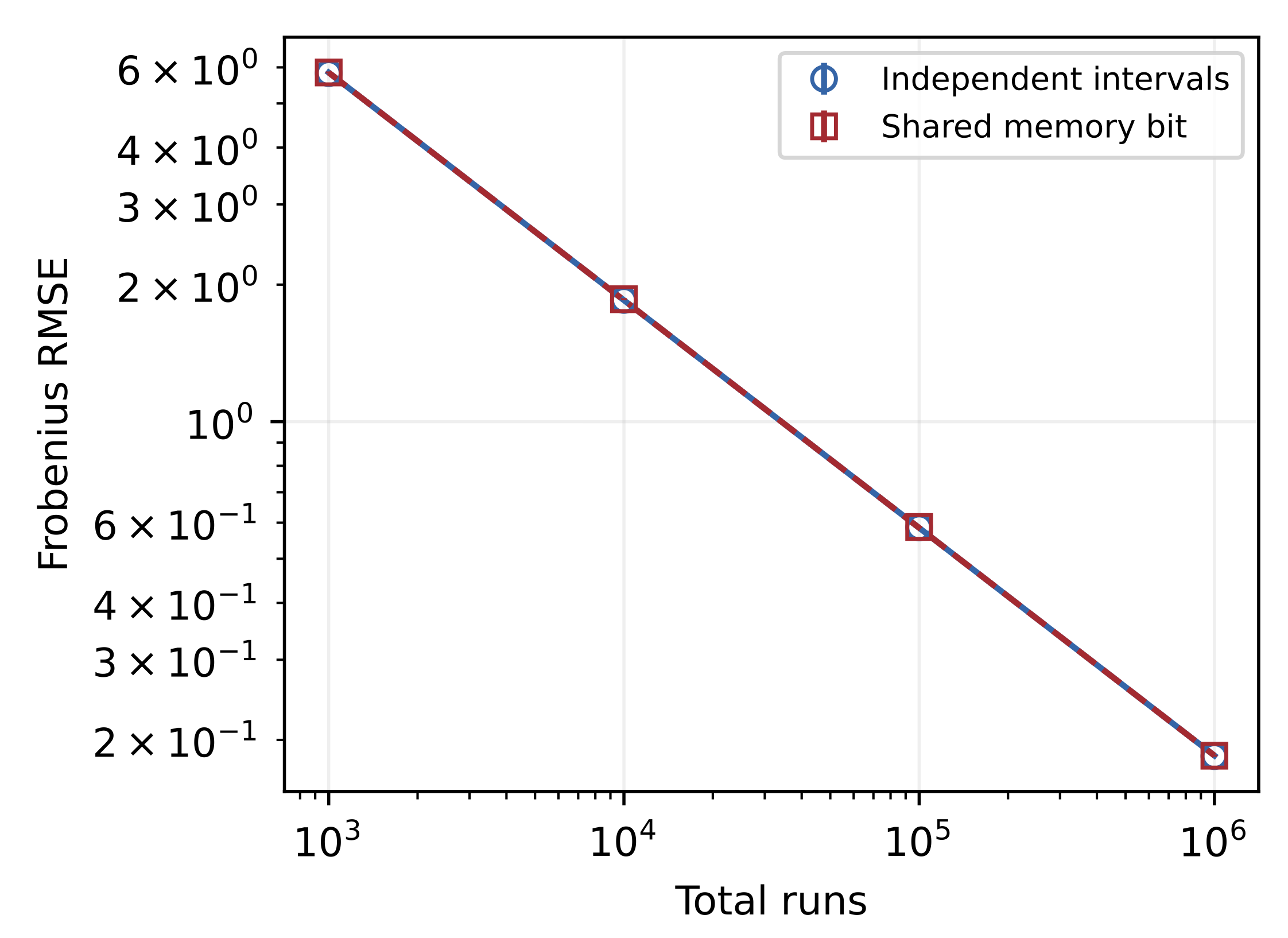}
    \caption{Three-time right-KD TST for independent intervals and a shared memory bit.}
    \label{fig:three-time-simu}
\end{figure}

We take $\rho=(I+0.6Z)/2$ and two evolutions as phase-flip  channel. In the memoryless model, each quantum channel independently applies $I$ or $Z$ with probability $1/2$. We consider two models with and without classical memory.
In the memory model, one fair classical bit is drawn at the beginning of each run and the same choice is used for both evolutions. Thus the individually averaged interval maps coincide, but the memory process cannot be replaced by their independent product. The classical bit is reset between runs, preserving the independent-run assumption of Theorem 2.

The record contains $16^2\times4=1024$ trajectories; the reconstructed temporal operator is $8\times8$. 
We independently compute its Pauli coefficients from the ordered probe expression and compare them with the exact snapshot inverse. The result are summarized in Figure~\ref{fig:three-time-simu} and Table~\ref{tab:newthree}.

\begin{table}[h]
\caption{Three-time Frobenius RMSE (400 trials). Parentheses give Monte Carlo standard errors.}\label{tab:newthree}
\centering\begin{tabular}{rcc}\hline
$N$ & Independent intervals & Shared bit\\\hline
$1000$ & 5.8298 (0.0299) & 5.8479 (0.0297) \\
$10000$ & 1.8511 (0.0101) & 1.8576 (0.0099) \\
$100000$ & 0.5861 (0.0032) & 0.5872 (0.0034) \\
$1000000$ & 0.1843 (0.0010) & 0.1847 (0.0010) \\
\hline\end{tabular}\end{table}

\begin{figure}[t]
\centering
\includegraphics[width=\linewidth]{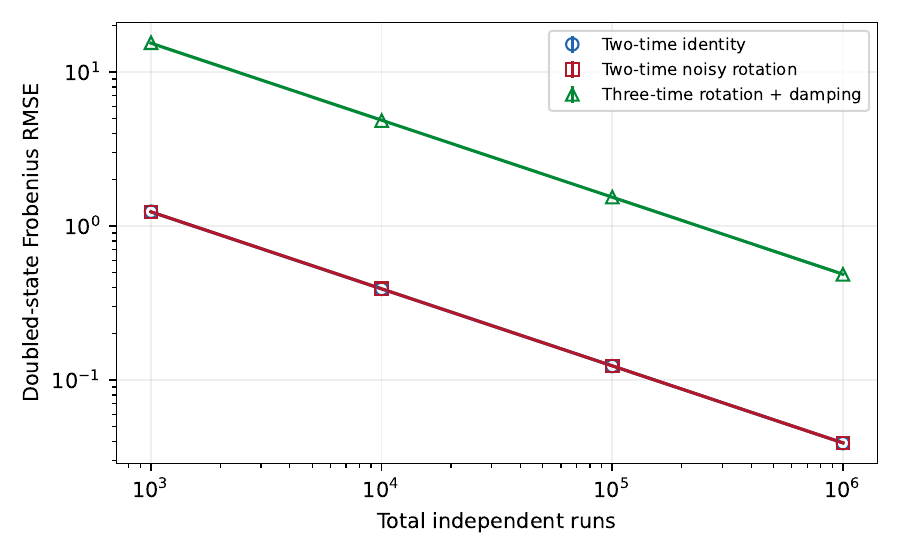}
\caption{Doubled-state reconstruction: empirical RMSE with one Monte Carlo standard error and exact variance predictions.}
\end{figure}

\subsection{Doubled Kirkwood-Dirac TST simulation}\label{app:doubled}

The doubled KD temporal state coincides with doubled density operator \cite{jia2024spatiotemporal,Jia2025TemporalKirkwoodDirac,jia2026unifiedspatiotemporalquantumstates}.
Write $m=n+1$, with $\operatorname{Tr}(\sigma_\mu\sigma_\nu)=d\delta_{\mu\nu}$ and $\sigma_0=I$. The doubled correlation tensor uniquely fixes
\begin{equation}
\overleftrightarrow{\Upsilon}=\frac{1}{d^{2m}}\sum_{\boldsymbol\mu,\boldsymbol\nu}
T_{\boldsymbol\mu;\boldsymbol\nu}\bigotimes_{k=n}^{0}
(\sigma_{\mu_k}\otimes\sigma_{\nu_k}).
\end{equation}
We use  $(L_n,R_n,\ldots,L_0,R_0)$ to denote the time slots, and the swap operator is
$\mathbb S=\sum_{i,j}|i\rangle\langle j|\otimes|j\rangle\langle i|$.
The doubled Jamio\l kowski operator is defined explicitly by \cite{jia2026unifiedspatiotemporalquantumstates}
\begin{equation}\begin{split}
\mathbb J[\mathcal E_k]=\sum_{i,j}
\big[(\mathcal E_k(|i\rangle\langle j|)\otimes I)\mathbb S\big]_{t_k}\\
\otimes\big[I\otimes|j\rangle\langle i|\big]_{t_{k-1}}.
\end{split}\end{equation} 
With the above convention, we have the recursive expression \cite{jia2026unifiedspatiotemporalquantumstates}
\begin{equation}
    \overleftrightarrow{\Upsilon}_k=(\mathbb J[\mathcal E_k]\otimes I_{\rm older})
       (I_{L_kR_k}\otimes \overleftrightarrow{\Upsilon}_{k-1}).
\end{equation}
See \cite{jia2026unifiedspatiotemporalquantumstates} for a detailed discussion.

We set initial state as $\rho=(I+0.24X-0.31Y+0.42Z)/2$. Three models are tested: (i) one identity channel; (ii) one noisy rotation $\mathcal E(A)=0.83UAU^\dagger+0.17\operatorname{Tr}(A)I/2$, with $U=\cos(0.37)I-i\sin(0.37)(X+Y+Z)/\sqrt3$; (iii) this rotation followed by amplitude damping with $\gamma=0.4$. 
The result is shown in Figure~\ref{fig:three-time-simu}.
The targets are $16\times16$, $16\times16$, and $64\times64$, respectively. The trajectory counts are 64, 64 and 1024.
At each $N\in\{10^3,10^4,10^5,10^6\}$ we generate 1000 datasets for each two-time model and 400 datasets for the three-time model. 

\end{revision}

\bibliographystyle{apsrev4-1-title}
\bibliography{Jiabib}

@article{Taranto2026HigherorderQ,
  title = {A Guide to Higher-Order Quantum Operations},
  author = {Taranto, Philip and Milz, Simon and Murao, Mio and Quintino, Marco T{\'u}lio and Modi, Kavan},
  journal = {PRX Quantum},
  volume = {7},
  issue = {4},
  pages = {042001},
  numpages = {128},
  year = {2026},
  month = {Nov},
  publisher = {American Physical Society},
  doi = {10.1103/92c4-g7qs},
  url = {https://link.aps.org/doi/10.1103/92c4-g7qs}
}

@article{Mele2026Optimal,
  author = {Mele, Antonio Anna and Bittel, Lennart},
  title = {Optimal learning of quantum channels in diamond distance},
  journal = {Nature Physics},
  year = {2026},
  doi = {10.1038/s41567-026-03499-7},
  url = {https://doi.org/10.1038/s41567-026-03499-7},
  isbn = {1745-2481}
}

@misc{bravoprieto2026quantummemoryadvantagequantum,
      title={Quantum memory advantage for quantum process tomography}, 
      author={Carlos Bravo-Prieto and Weiyuan Gong and Antonio Anna Mele},
      year={2026},
      eprint={2607.13476},
      archivePrefix={arXiv},
      primaryClass={quant-ph},
      url={https://arxiv.org/abs/2607.13476}, 
}

@article{SurawyStepney2022projectedleast,
  doi = {10.22331/q-2022-10-20-844},
  url = {https://doi.org/10.22331/q-2022-10-20-844},
  title = {Projected {L}east-{S}quares {Q}uantum {P}rocess {T}omography},
  author = {Surawy-Stepney, Trystan and Kahn, Jonas and Kueng, Richard and Guta, Madalin},
  journal = {{Quantum}},
  issn = {2521-327X},
  publisher = {{Verein zur F{\"{o}}rderung des Open Access Publizierens in den Quantenwissenschaften}},
  volume = {6},
  pages = {844},
  month = oct,
  year = {2022}
}

@ARTICLE{Haah2017QST,
  author={Haah, Jeongwan and Harrow, Aram W. and Ji, Zhengfeng and Wu, Xiaodi and Yu, Nengkun},
  journal={IEEE Transactions on Information Theory}, 
  title={Sample-Optimal Tomography of Quantum States}, 
  year={2017},
  volume={63},
  number={9},
  pages={5628-5641},
  doi={10.1109/TIT.2017.2719044}}

@inproceedings{ODonnell2016QST,
author = {O'Donnell, Ryan and Wright, John},
title = {Efficient quantum tomography},
year = {2016},
isbn = {9781450341325},
publisher = {Association for Computing Machinery},
address = {New York, NY, USA},
url = {https://doi.org/10.1145/2897518.2897544},
doi = {10.1145/2897518.2897544},
booktitle = {Proceedings of the Forty-Eighth Annual ACM Symposium on Theory of Computing},
pages = {899–912},
numpages = {14},
location = {Cambridge, MA, USA},
series = {STOC '16}
}

@misc{jia2026unifiedspatiotemporalquantumstates,
      title={Unified spatiotemporal quantum states and spatiotemporal entanglement from Kirkwood-Dirac phase space}, 
      author={Zhian Jia and Mei-Hui Xiao},
      year={2026},
      eprint={2609.24680},
      archivePrefix={arXiv},
      primaryClass={quant-ph},
      url={https://arxiv.org/abs/2609.24680}, 
}

@article{ding2026quantifying,
  title={Quantifying the Nonclassicality of the Kirkwood--Dirac Quasiprobability Distribution Under Discrete-Time Dynamics},
  author={Ding, Ziheng and Zhou, Si-Qi},
  journal={Entropy},
  volume={28},
  number={4},
  pages={395},
  year={2026},
  url={https://www.mdpi.com/1099-4300/28/4/395}
}

@article{Li2025combtomography,
  title = {Quantum Comb Tomography via Learning Isometries on Stiefel Manifold},
  author = {Li, Ze-Tong and He, Xin-Lin and Zheng, Cong-Cong and Dong, Yu-Qian and Luan, Tian and Yu, Xu-Tao and Zhang, Zai-Chen},
  journal = {Phys. Rev. Lett.},
  volume = {134},
  issue = {1},
  pages = {010803},
  numpages = {7},
  year = {2025},
  month = {Jan},
  publisher = {American Physical Society},
  doi = {10.1103/PhysRevLett.134.010803},
  url = {https://link.aps.org/doi/10.1103/PhysRevLett.134.010803},
  eprint={2404.06988},
  archivePrefix={arXiv},
  primaryClass={quant-ph} 
}

@article{Antesberger2024highertomography,
  title = {Higher-Order Process Matrix Tomography of a Passively-Stable Quantum Switch},
  author = {Antesberger, Michael and Quintino, Marco T\'ulio and Walther, Philip and Rozema, Lee A.},
  journal = {PRX Quantum},
  volume = {5},
  issue = {1},
  pages = {010325},
  numpages = {22},
  year = {2024},
  month = {Feb},
  publisher = {American Physical Society},
  doi = {10.1103/PRXQuantum.5.010325},
  url = {https://link.aps.org/doi/10.1103/PRXQuantum.5.010325},
          eprint={2305.19386 },
  archivePrefix={arXiv},
  primaryClass={quant-ph} 
}

@article{Li2020hamiltoniantom,
  title = {Hamiltonian Tomography via Quantum Quench},
  author = {Li, Zhi and Zou, Liujun and Hsieh, Timothy H.},
  journal = {Phys. Rev. Lett.},
  volume = {124},
  issue = {16},
  pages = {160502},
  numpages = {6},
  year = {2020},
  month = {Apr},
  publisher = {American Physical Society},
  doi = {10.1103/PhysRevLett.124.160502},
  url = {https://link.aps.org/doi/10.1103/PhysRevLett.124.160502},
        eprint={1912.09492},
  archivePrefix={arXiv},
  primaryClass={quant-ph} 
}

@article{Mohseni2008processtom,
  title = {Quantum-process tomography: Resource analysis of different strategies},
  author = {Mohseni, M. and Rezakhani, A. T. and Lidar, D. A.},
  journal = {Phys. Rev. A},
  volume = {77},
  issue = {3},
  pages = {032322},
  numpages = {15},
  year = {2008},
  month = {Mar},
  publisher = {American Physical Society},
  doi = {10.1103/PhysRevA.77.032322},
  url = {https://link.aps.org/doi/10.1103/PhysRevA.77.032322},
      eprint={quant-ph/0702131},
  archivePrefix={arXiv},
  primaryClass={quant-ph} 
}

@article{White2022processtom,
  title = {Non-Markovian Quantum Process Tomography},
  author = {White, G.A.L. and Pollock, F.A. and Hollenberg, L.C.L. and Modi, K. and Hill, C.D.},
  journal = {PRX Quantum},
  volume = {3},
  issue = {2},
  pages = {020344},
  numpages = {30},
  year = {2022},
  month = {May},
  publisher = {American Physical Society},
  doi = {10.1103/PRXQuantum.3.020344},
  url = {https://link.aps.org/doi/10.1103/PRXQuantum.3.020344},
    eprint={2106.11722},
  archivePrefix={arXiv},
  primaryClass={quant-ph} 
}

@Article{Anshu2024,
author={Anshu, Anurag
and Arunachalam, Srinivasan},
title={A survey on the complexity of learning quantum states},
journal={Nature Reviews Physics},
year={2024},
month={Jan},
day={01},
volume={6},
number={1},
pages={59-69},
issn={2522-5820},
doi={10.1038/s42254-023-00662-4},
url={https://doi.org/10.1038/s42254-023-00662-4},
  eprint={2305.20069},
  archivePrefix={arXiv},
  primaryClass={quant-ph} 
}

@misc{Kelvin2025KD,
      title={Quantum Stochastic Processes with Kirkwood-Dirac Quasiprobabilities: Nonclassicality, Non-Markovianity, Simulation and Quantum Computing Advantage}, 
      author={Kelvin Onggadinata and Arghya Maity and Zhian Jia and John F. Kam and Felix C. Binder and Kavan Modi and Teck Seng Koh},
      year={2025},
      note={In preparation}
}

@misc{lie2025probingquantumstatesspacetime,
  title = {Probing Quantum States over Spacetime through Interferometry},
  author = {Lie, Seok Hyung and Kwon, Hyukjoon},
  journal = {Phys. Rev. Lett.},
  volume = {136},
  issue = {25},
  pages = {250201},
  numpages = {10},
  year = {2026},
  month = {Jun},
  publisher = {American Physical Society},
  doi = {10.1103/t6jg-z17r},
  url = {https://link.aps.org/doi/10.1103/t6jg-z17r},
   eprint={2507.19258},
  archivePrefix={arXiv},
  primaryClass={quant-ph} 
}

@article{Lie2025stateovertime,
  title = {Multipartite Quantum States over Time from Two Fundamental Assumptions},
  author = {Lie, Seok Hyung and Fullwood, James},
  journal = {Phys. Rev. Lett.},
  volume = {135},
  issue = {23},
  pages = {230204},
  numpages = {8},
  year = {2025},
  month = {Dec},
  publisher = {American Physical Society},
  doi = {10.1103/lbf3-snp8},
  url = {https://link.aps.org/doi/10.1103/lbf3-snp8}
}

@misc{Jia2025TemporalKirkwoodDirac,
   title={Temporal Kirkwood-Dirac Quasiprobability Distribution and Unification of Temporal State Formalisms through Temporal Bloch Tomography}, 
      author={Zhian Jia and Kavan Modi and Dagomir Kaszlikowski},
      year={2026},
      eprint={2601.05294},
      archivePrefix={arXiv},
      primaryClass={quant-ph},
      url={https://arxiv.org/abs/2601.05294}, 
}

@Article{Wang2024KD,
author={Wang, Pengfei
and Kwon, Hyukjoon
and Luan, Chun-Yang
and Chen, Wentao
and Qiao, Mu
and Zhou, Zinan
and Wang, Kaizhao
and Kim, M. S.
and Kim, Kihwan},
title={Snapshotting quantum dynamics at multiple time points},
journal={Nature Communications},
year={2024},
month={Oct},
day={16},
volume={15},
number={1},
pages={8900},
issn={2041-1723},
doi={10.1038/s41467-024-53051-5},
url={https://doi.org/10.1038/s41467-024-53051-5},
  eprint={2207.06106},
  archivePrefix={arXiv},
  primaryClass={quant-ph} 
}

@article{Gherardini2024KD,
  title = {Quasiprobabilities in Quantum Thermodynamics and Many-Body Systems},
  author = {Gherardini, Stefano and De Chiara, Gabriele},
  journal = {PRX Quantum},
  volume = {5},
  issue = {3},
  pages = {030201},
  numpages = {38},
  year = {2024},
  month = {Sep},
  publisher = {American Physical Society},
  doi = {10.1103/PRXQuantum.5.030201},
  url = {https://link.aps.org/doi/10.1103/PRXQuantum.5.030201},
  eprint={2403.17138},
  archivePrefix={arXiv},
  primaryClass={quant-ph} 
}

@misc{fullwood2025spatiotemporalbornrule,
      title={The spatiotemporal Born rule is quasiprobabilistic}, 
      author={James Fullwood and Zhihao Ma and Zhen Wu},
      year={2025},
      eprint={2507.16919},
      archivePrefix={arXiv},
      primaryClass={quant-ph},
      url={https://arxiv.org/abs/2507.16919}, 
}

@article{margenau1961correlation,
  title={Correlation between measurements in quantum theory},
  author={Margenau, Henry and Hill, Robert Nyden},
  journal={Progress of Theoretical Physics},
  volume={26},
  number={5},
  pages={722--738},
  year={1961},
  publisher={Oxford University Press},
  url={https://academic.oup.com/ptp/article/26/5/722/1936017},
  doi={10.1143/PTP.26.722},
}

@article{jia2024spatiotemporal,
  title={The Spatiotemporal Doubled-Density Operator: A Unified Framework for Analyzing Spatial and Temporal Quantum Processes},
  author={Jia, Zhian and Kaszlikowski, Dagomir},
  journal={Advanced Quantum Technologies},
  volume={7},
  number={11},
  pages={2400102},
  year={2024},
  publisher={Wiley Online Library},
  url={https://advanced.onlinelibrary.wiley.com/doi/10.1002/qute.202400102},
  doi = {10.1002/qute.202400102},
  eprint={2305.15649},
  archivePrefix={arXiv},
  primaryClass={quant-ph}   
}

@article{Chiribella2009comb,
  title = {Theoretical framework for quantum networks},
  author = {Chiribella, Giulio and D'Ariano, Giacomo Mauro and Perinotti, Paolo},
  journal = {Phys. Rev. A},
  volume = {80},
  issue = {2},
  pages = {022339},
  numpages = {20},
  year = {2009},
  month = {Aug},
  publisher = {American Physical Society},
  doi = {10.1103/PhysRevA.80.022339},
  url = {https://link.aps.org/doi/10.1103/PhysRevA.80.022339},
  eprint={0904.4483},
  archivePrefix={arXiv},
  primaryClass={quant-ph}   
}

@article{Aharonov2009multi,
  title = {Multiple-time states and multiple-time measurements in quantum mechanics},
  author = {Aharonov, Yakir and Popescu, Sandu and Tollaksen, Jeff and Vaidman, Lev},
  journal = {Phys. Rev. A},
  volume = {79},
  issue = {5},
  pages = {052110},
  numpages = {16},
  year = {2009},
  month = {May},
  publisher = {American Physical Society},
  doi = {10.1103/PhysRevA.79.052110},
  url = {https://link.aps.org/doi/10.1103/PhysRevA.79.052110},
  eprint={0712.0320},
  archivePrefix={arXiv},
  primaryClass={quant-ph}    
}

@article{Parzygnat2023pdo,
  title = {From Time-Reversal Symmetry to Quantum Bayes' Rules},
  author = {Parzygnat, Arthur J. and Fullwood, James},
  journal = {PRX Quantum},
  volume = {4},
  issue = {2},
  pages = {020334},
  numpages = {26},
  year = {2023},
  month = {Jun},
  publisher = {American Physical Society},
  doi = {10.1103/PRXQuantum.4.020334},
  url = {https://link.aps.org/doi/10.1103/PRXQuantum.4.020334},
eprint={2212.08088},
  archivePrefix={arXiv},
  primaryClass={quant-ph}  
}

@Article{Liu2025PDO,
author={Liu, Xiangjing
and Qiu, Yixian
and Dahlsten, Oscar
and Vedral, Vlatko},
title={Quantum causal inference with extremely light touch},
journal={npj Quantum Information},
year={2025},
month={Mar},
day={29},
volume={11},
number={1},
pages={54},
issn={2056-6387},
doi={10.1038/s41534-024-00956-0},
url={https://doi.org/10.1038/s41534-024-00956-0},  
eprint={2303.10544},
  archivePrefix={arXiv},
  primaryClass={quant-ph}  
}

@article{Lostaglio2023kirkwooddirac,
  doi = {10.22331/q-2023-10-09-1128},
  url = {https://doi.org/10.22331/q-2023-10-09-1128},
  title = {Kirkwood-{D}irac quasiprobability approach to the statistics of incompatible observables},
  author = {Lostaglio, Matteo and Belenchia, Alessio and Levy, Amikam and Hern{\'{a}}ndez-G{\'{o}}mez, Santiago and Fabbri, Nicole and Gherardini, Stefano},
  journal = {{Quantum}},
  issn = {2521-327X},
  publisher = {{Verein zur F{\"{o}}rderung des Open Access Publizierens in den Quantenwissenschaften}},
  volume = {7},
  pages = {1128},
  month = oct,
  year = {2023},
   eprint={2206.11783 },
  archivePrefix={arXiv},
  primaryClass={quant-ph}   
}

@article{ArvidssonShukur2024KDreview,
doi = {10.1088/1367-2630/ada05d},
url = {https://dx.doi.org/10.1088/1367-2630/ada05d},
year = {2024},
month = {dec},
publisher = {IOP Publishing},
volume = {26},
number = {12},
pages = {121201},
author = {Arvidsson-Shukur, David R M and Braasch Jr, William F and De Bièvre, Stephan and Dressel, Justin and Jordan, Andrew N and Langrenez, Christopher and Lostaglio, Matteo and Lundeen, Jeff S and Halpern, Nicole Yunger},
title = {Properties and applications of the Kirkwood–Dirac distribution},
journal = {New Journal of Physics},
   eprint={2403.18899},
  archivePrefix={arXiv},
  primaryClass={quant-ph}   
}

@article{Song2024causal,
  title = {Causal Classification of Spatiotemporal Quantum Correlations},
  author = {Song, Minjeong and Narasimhachar, Varun and Regula, Bartosz and Elliott, Thomas J. and Gu, Mile},
  journal = {Phys. Rev. Lett.},
  volume = {133},
  issue = {11},
  pages = {110202},
  numpages = {6},
  year = {2024},
  month = {Sep},
  publisher = {American Physical Society},
  doi = {10.1103/PhysRevLett.133.110202},
  url = {https://link.aps.org/doi/10.1103/PhysRevLett.133.110202},
   eprint={2306.09336 },
  archivePrefix={arXiv},
  primaryClass={quant-ph}   
}

@article {liu2023unification,
    AUTHOR = {Liu, Xiangjing and Jia, Zhian and Qiu, Yixian and Li, Fei and
              Dahlsten, Oscar},
     TITLE = {Unification of spatiotemporal quantum formalisms: mapping
              between process and pseudo-density matrices via multiple-time
              states},
   JOURNAL = {New J. Phys.},
  FJOURNAL = {New Journal of Physics},
    VOLUME = {26},
      YEAR = {2024},
    NUMBER = {March},
     PAGES = {Paper No. 033008, 15},
   MRCLASS = {81P16},
  MRNUMBER = {4730855},
       DOI = {10.1088/1367-2630/ad264c},
       URL = {https://doi.org/10.1088/1367-2630/ad264c},
    eprint={2306.05958},
  archivePrefix={arXiv},
  primaryClass={quant-ph}   
}

@article{fullwood2023quantum,
  title={Quantum dynamics as a pseudo-density matrix},
  author={Fullwood, James},
  journal={Quantum},
  volume={9},
  pages={1719},
  year={2025},
  url={https://doi.org/10.22331/q-2025-04-24-1719},
  eprint={2304.03954},
  archivePrefix={arXiv},
  primaryClass={quant-ph}      
}

@article{fullwood2022quantum,
  title={On quantum states over time},
  author={Fullwood, James and Parzygnat, Arthur J},
  journal={Proceedings of the Royal Society A},
  volume={478},
  number={2264},
  pages={20220104},
  year={2022},
  publisher={The Royal Society},
  url={https://royalsocietypublishing.org/doi/abs/10.1098/rspa.2022.0104},
    eprint={2202.03607},
  archivePrefix={arXiv},
  primaryClass={quant-ph}   
}

@article{Dirac1945on,
  title = {On the Analogy Between Classical and Quantum Mechanics},
  author = {Dirac, P. A. M.},
  journal = {Rev. Mod. Phys.},
  volume = {17},
  issue = {2-3},
  pages = {195--199},
  numpages = {0},
  year = {1945},
  month = {Apr},
  publisher = {American Physical Society},
  doi = {10.1103/RevModPhys.17.195},
  url = {https://link.aps.org/doi/10.1103/RevModPhys.17.195}
}

@article{Kirkwood1933quantum,
  title = {Quantum Statistics of Almost Classical Assemblies},
  author = {Kirkwood, John G.},
  journal = {Phys. Rev.},
  volume = {44},
  issue = {1},
  pages = {31--37},
  numpages = {0},
  year = {1933},
  month = {Jul},
  publisher = {American Physical Society},
  doi = {10.1103/PhysRev.44.31},
  url = {https://link.aps.org/doi/10.1103/PhysRev.44.31}
}

@inproceedings{gutoski2007toward,
 title={Toward a general theory of quantum games},
  author={Gutoski, Gus and Watrous, John},
  booktitle={Proceedings of the thirty-ninth annual ACM symposium on Theory of computing},
  pages={565--574},
  year={2007},
  url={https://dl.acm.org/doi/10.1145/1250790.1250873},
  eprint={quant-ph/0611234},
  archivePrefix={arXiv},
  primaryClass={quant-ph}    
}

@article{cotler2018superdensity,
  title={Superdensity operators for spacetime quantum mechanics},
  author={Cotler, Jordan and Jian, Chao-Ming and Qi, Xiao-Liang and Wilczek, Frank},
  journal={Journal of High Energy Physics},
  volume={2018},
  number={9},
  pages={1--57},
  year={2018},
  publisher={Springer},
  url={https://link.springer.com/article/10.1007/JHEP09(2018)093},
  eprint={1711.03119 },
  archivePrefix={arXiv},
  primaryClass={quant-ph}       
}

@article{griffiths1984consistent,
  title={Consistent histories and the interpretation of quantum mechanics},
  author={Griffiths, Robert B},
  journal={Journal of Statistical Physics},
  volume={36},
  number={1},
  pages={219--272},
  year={1984},
  publisher={Springer},
  url={https://link.springer.com/article/10.1007/BF01015734}
}

@article{oreshkov2012quantum,
  title={Quantum correlations with no causal order},
  author={Oreshkov, Ognyan and Costa, Fabio and Brukner, {\v{C}}aslav},
  journal={Nature Communications},
  volume={3},
  number={1},
  pages={1--8},
  year={2012},
  publisher={Nature Publishing Group},
  url={https://www.nature.com/articles/ncomms2076},
  eprint={1105.4464},
  archivePrefix={arXiv},
  primaryClass={quant-ph}     
}

@article{Wigner1932on,
  title = {On the Quantum Correction For Thermodynamic Equilibrium},
  author = {Wigner, E.},
  journal = {Phys. Rev.},
  volume = {40},
  issue = {5},
  pages = {749--759},
  numpages = {0},
  year = {1932},
  month = {Jun},
  publisher = {American Physical Society},
  doi = {10.1103/PhysRev.40.749},
  url = {https://link.aps.org/doi/10.1103/PhysRev.40.749}
}

@article{fitzsimons2015quantum,
  title={Quantum correlations which imply causation},
  author={Fitzsimons, Joseph F and Jones, Jonathan A and Vedral, Vlatko},
  journal={Scientific Reports},
  volume={5},
  number={1},
  pages={1--7},
  year={2015},
  publisher={Nature Publishing Group},
  url={https://www.nature.com/articles/srep18281},
  eprint={1302.2731},
  archivePrefix={arXiv},
  primaryClass={quant-ph} 
}

@article{Pollock2018processtensor,
  title = {Non-Markovian quantum processes: Complete framework and efficient characterization},
  author = {Pollock, Felix A. and Rodr\'{\i}guez-Rosario, C\'esar and Frauenheim, Thomas and Paternostro, Mauro and Modi, Kavan},
  journal = {Phys. Rev. A},
  volume = {97},
  issue = {1},
  pages = {012127},
  numpages = {13},
  year = {2018},
  month = {Jan},
  publisher = {American Physical Society},
  doi = {10.1103/PhysRevA.97.012127},
  url = {https://link.aps.org/doi/10.1103/PhysRevA.97.012127}
}

\end{document}